\documentclass[11pt,a4paper,fleqn]{article}

\usepackage{amsthm,amsmath}
\usepackage{amssymb,amstext,amsbsy,bbm,natbib}
\usepackage{color}
\definecolor{darkred}{rgb}{0.6,0.0,0.1}
\definecolor{darkgreen}{rgb}{0,0.5,0}
\definecolor{darkblue}{rgb}{0,0,0.5}

\def\Diag{\mathop{\rm Diag}\nolimits}%

\newcommand{\Gw}{\upsilon}
\newcommand{\Lw}{\Lambda}
\newcommand{\Gd}{d}
\newcommand{\kstar}{m_*}
\newcommand{\ukstar}{\underline{\kstar}}
\newcommand{\dstar}{\delta_n^*}
\newcommand{\bw}{b}
\newcommand{\br}{\rho}
\newcommand{\hw}{\omega}

\def\cov{\mathop{\rm cov}\nolimits}%
\def\Ex{{\mathbb E}}%
\newcommand{\skalarV}[1]{\langle #1\rangle}   %Skalarprodukt
\DeclareMathOperator{\norm}{\lVert\cdot\rVert}   %Norm
\newcommand{\normV}[1]{\lVert#1\rVert}   %Norm
\newcommand{\N}{{\mathbb N}}
\newcommand{\R}{{\mathbb R}}

\newcommand{\Z}{{\mathbb Z}}
\newcommand{\1}{{\mathbbm 1}}

\newcommand{\cE}{{\cal E}}

\newcommand{\cN}{{\cal N}}

\newcommand{\cR}{{\cal R}}

\newcommand{\cW}{{\cal W}}
\newcommand{\cX}{{\cal X}}

\newcommand{\um}{\underline{m}}

\usepackage{verbatim}
\usepackage{ bm}
\definecolor{dgreen}{rgb}{0,0.5,0}
\definecolor{dblue}{rgb}{0,0,0.9}
\definecolor{dred}{rgb}{0.6,0.0,0.1}
\definecolor{dgold}{rgb}{0.5,0.3,0.0}
\definecolor{dvio}{rgb}{0.6,0.3,0.5}
\definecolor{gray}{rgb}{0.5,0.5,0.5}

\newcommand{\dr}{\color{dred}}

\newtheoremstyle{mysc}% name
  {3pt}%      Space above
  {3pt}%      Space below
  {\it}%         Body font
  {}%         Indent amount (empty = no indent, \parindent = para indent)
  {\color{darkred}\sc}% Thm head font
  {.}%        Punctuation after thm head
  {.5em}%     Space after thm head: " " = normal interword space;
  {}%         Thm head spec (can be left empty, meaning `normal')

\newtheoremstyle{myex}% name
  {10pt}%      Space above
  {10pt}%      Space below
  {\rm}%         Body font
  {}%         Indent amount (empty = no indent, \parindent = para indent)
  {\color{darkred}\sc}% Thm head font
  {.}%        Punctuation after thm head
  {.5em}%     Space after thm head: " " = normal interword space;
  {}%         Thm head spec (can be left empty, meaning `normal')

\theoremstyle{mysc}\newtheorem{prop}{Proposition}[section]
\theoremstyle{mysc}\newtheorem{assumption}{Assumption}[section]
\theoremstyle{mysc}\newtheorem{coro}[prop]{Corollary}
\theoremstyle{mysc}\newtheorem{theo}[prop]{Theorem}
\theoremstyle{mysc}
\theoremstyle{mysc}\newtheorem{lem}[prop]{Lemma}
\theoremstyle{myex}\newtheorem{rem}{Remark}[section]
\theoremstyle{myex}
\theoremstyle{myex}

\numberwithin{equation}{section}

\author{{\sc Herv\'e Cardot}\thanks{Universit\'e de Bourgogne, Institut de Math\'ematiques de Bourgogne, 9 Av. Alain Savary, 21078 Dijon Cedex, France, e-mail: \texttt{herve.cardot@u-bourgogne.fr}}
	\and{\sc Jan Johannes}\thanks{Universit\"at Heidelberg, Institut f\"ur Angewandte Mathematik, Im Neuenheimer Feld, 294, D-69120 Heidelberg, Germany, e-mail: \texttt{johannes@statlab.uni-heidelberg.de}}}
\title{{\bf Thresholding Projection Estimators in Functional Linear Models.}}

\begin{document}
\date{December 17, 2008}
\maketitle

\begin{abstract} %
We consider the problem of estimating the regression function in  functional linear regression models by proposing a new type of  projection estimators which combine dimension reduction and thresholding. The introduction of a threshold rule allows to get consistency under broad assumptions as well as minimax rates of convergence under additional regularity hypotheses. We also consider the particular case of Sobolev spaces generated by the trigonometric basis which permits to get easily mean squared error of prediction as well as estimators of the derivatives of the regression function. We prove these estimators are   minimax and rates of convergence are given for some particular cases. 

\end{abstract}
%\vspace{0.5cm}
\begin{tabbing}
\>Mean squared error of prediction, Optimal rate of convergence, \\
\>Hilbert scale, Sobolev Space.\\[.2ex]
\end{tabbing}

\section{Introduction}\label{sec:intro}
Functional data analysis (\cite{RamsaySilverman2005}, \cite{FerratyVieu2006}) is a topic of growing interest in statistics and many applications in chemometrics (\cite{FrankFriedman93}), finance (\cite{PredaSaporta2005}), biometry or climatology (\cite{BesseCardotStephenson2000}) are now dealing with the functional linear model. This model is useful to estimate or predict a scalar random variable, say $Y \in \R$, thanks to a random function denoted by $X.$
%\note{Sentence is to long.} 
We assume in the following that  $Y$ and $X$ are centered random variables and, without loss of generality, that the random function $X$ takes  values  in $L^2[0,1]$, the space of square integrable functions defined on $[0,1]$ endowed with its usual inner product $\skalarV{f,g}=\int_0^1f(t){g(t)}dt$ and associated norm $\normV{f}=\skalarV{f,f}^{1/2},$ $f,g\in L^2[0,1].$ The functional linear model is then defined by
\begin{equation}\label{intro:e1}
Y=\int_{0}^1\beta(t)X(t)dt +\sigma\epsilon,\quad\sigma>0,
\end{equation}
where the function $\beta(t)$ is called the regression or slope function and the error term $\epsilon$ is supposed to be centered $\Ex (\epsilon)=0$ and  not correlated with $X$: $\forall \ t \in [0,1], \ \Ex (X(t)\epsilon)=0.$

Assuming that $X$ has a finite second moment, i.e. $\Ex \normV{X}^2 = \int_0^1\Ex|X(t)|^2dt <\infty,$  one can define the covariance operator of $X,$ say $\Gamma.$ This operator is defined on $L^2[0,1]$ as follows: for any function $f \in L^2[0,1],$ 
\begin{equation}\label{cov-op}
\Gamma f(s)  = \int_0^1 \cov(X(t),X(s)) f(t) \ dt, \ \forall s \in [0,1]. 
\end{equation}
It is well known (see \textit{e.g.} \cite{CardotFerratySarda99}) that the regression function $\beta$ satisfies the following moment equation
\begin{equation}
\label{method:e1}g(s):=\Ex[YX(s)]=[\Gamma \beta](s),\quad s\in[0,1],
\end{equation}
where $g$ belongs to $L^2[0,1].$  Since $\Gamma$ is a non negative nuclear operator (\cite{DauxoisPousseRomain82}) a continuous generalized inverse of $\Gamma$ does not exist as long as the range of the operator $\Gamma$ is an infinite dimensional subspace of $L^2[0,1].$ Consequently  inverting equation (\ref{method:e1}) to recover $\beta$ can be seen as an ill posed inverse problem.
%\note{Sentence is too long.}  
\cite{CardotFerratySarda2003}  provides   a   necessary and sufficient condition for the existence of a unique solution of equation (\ref{method:e1}) 
\begin{assumption}\label{ass:ident} 
The covariance operator $\Gamma$ of the random function $X$ is injective and the function $g=\Ex[YX]$ belongs to the range $\cR(\Gamma)$ of $\Gamma$.
\end{assumption} 
Under this assumption, the covariance operator $\Gamma$  admits a discrete spectral decomposition given by a  sequence $(\lambda_j)_{j\in\N}$ of strictly positive eigenvalues %ordered in decreasing order\note{we do not need this} 
and a sequence of corresponding orthonormal eigenfunctions $\{\phi_j\}_{j\in\N}.$ Then,  the normal equation \eqref{method:e1} can be rewritten as follows%
\begin{equation}\label{form:e1}
\beta = \sum_{j\in\N} \frac{g_j}{\lambda_j}\cdot \phi_j\quad\mbox{ with }g_j:=\skalarV{g,\phi_j},\; j\in\N . %\;\mbox{ and } g:=\Ex [YX].
\end{equation}
It is well-known that, even in case of a-priori known  eigenvalues $\{\lambda_j\}$ and eigenfunctions $\{\phi_j\},$ replacing in (\ref{form:e1}) the unknown function $g$ by a consistent estimator $\widehat{g}$ does in general not lead  to a consistent estimator  of $\beta$. To be more precise, since the sequence $(\lambda_j)_{j\in\N}$ tends to zero, $\Ex\normV{\widehat{g}- g}^2=o(1)$ does generally not imply $\sum_{j\in\N} |\lambda_j|^{-2}\cdot\Ex|\skalarV{\widehat{g}-g,\phi_j}|^2 =o(1)$. Consequently, the estimation in functional linear model is called ill-posed and  additional regularity assumptions on the regression function $\beta$ are necessary in order to obtain a uniform rate of convergence (c.f. \cite{EHN00}).

The objective is to estimate the regression function $\beta, $ as well as its derivatives, when observing a sample $(Y_i,X_i)$ of $n$ i.i.d realizations of $(Y,X).$ We can define the empirical estimators of  $g$ and $\Gamma$ respectively as follows
%the covariance operators 
\begin{equation}\label{set:e3}
\widehat{g}:=\frac{1}{n}\sum_{i=1}^n Y_iX_i \quad \mbox{and} \quad \widehat{\Gamma}:=\frac{1}{n}\sum_{i=1}^n \skalarV{X_i,\cdot}X_i \ .
\end{equation}

The main class of  estimation procedures studied in the statistical literature are based on principal components regression and consist in reducing the dimension by inverting equation (\ref{method:e1}) in the finite dimension space generated by the eigenfunctions of $\widehat{\Gamma}$ associated to the  largest eigenvalues  (see \textit{e.g.} \cite{Bosq2000}, \cite{FrankFriedman93}, \cite{CardotFerratySarda99},  \cite{CardotMasSarda2007} or \cite{MullerStadtmuller2005} in the context of generalized linear models). 
\begin{comment}
Recently, \cite{HallHorowitz2007} have proved under some assumptions on the eigenvalues of the covariance operator and the regularity of function $\beta$ that it  leads, provided the dimension tends with an appropriate rate to infinity, to estimators with an optimal rate of convergence according to the mean quadratic error criterion $\Ex \normV{\widehat{\beta} - \beta}^2.$ \end{comment}

The second important class of estimators relies on minimizing a penalized least squares criterion  which can be seen as  generalization of the ridge regression. 
\cite{MarxEilers99} and \cite{CardotFerratySarda2003} proposed B-splines expansion of the regression function with a penalty dealing with the squared norm of a fixed order derivative of the estimators. More recently \cite{CrambesKneipSarda2007} proposed a spline smoothing decomposition with the same type of penalty and proved the optimality of their  estimators according to a criterion that can be interpreted as a squared error of prediction. Note that this question has given rise recently to numerous publications in the machine learning community with similar ideas  based on reproducing kernel Hilbert spaces (RKHS) and Tikhonov regularization (see e.g. \cite{SmaleZhou07}, \cite{BauerPR07} and references therein).

Borrowing ideas from the inverse problems community (\cite{EfromovichKoltchinskii2001} and \cite{HoffmannReiss2008}) we propose in this article a new class of estimators which rely on dimension reduction  by projecting the data onto some basis of orthonormal functions and threshold techniques that allow to control the accuracy of the estimator. More precisely, let us consider a set of orthonormal functions such as wavelet or trigonometric basis denoted by  $\{\psi_1,\dotsc,\psi_m, \dotsc\}$ which forms a basis of $L^2[0,1].$
Given a dimension $m\geq 1,$ 
we denote by $[\widehat{\Gamma}]_{\um}$ the $m\times m$ matrix with generic elements $\skalarV{\widehat{\Gamma} \psi_{\ell},\psi_{j}}, j,\ell=1, \ldots, m$ and by $[\widehat{g}]_{\um}$ the $m$ vector with elements $\skalarV{\widehat{g},\psi_{\ell}}, \ell=1, \ldots, m.$ We can first remark, that the least squares estimator of $\beta$ obtained with  the projections of the  $X_i$ onto $\Psi_m,$  the subspace of $L^2[0,1]$ spanned by the functions $\{\psi_1,\dotsc,\psi_m\},$ is simply given, when $[\widehat{\Gamma}]_{\um}$ is non singular, by $([\widehat{\Gamma}]_{\um}^{-1} [\widehat{g}]_{\um})^t  [\psi]_{\um}(\cdot)$  where $[\psi]_{\um}(\cdot) = (\psi_{1}(\cdot), \ldots, \psi_{m}(\cdot))^t.$
Our estimator, in its simplest form, consists in thresholding this projection estimator when, roughly speaking, the norm of the inverse of the matrix $[\widehat{\Gamma}]_{\um}$ is too large. More precisely, introducing a threshold value $\gamma$ which will depend on  $m$ and $n$ we propose to estimate $\beta$ as follows
\begin{equation}\label{beta:s0}
\widehat{\beta}(t)  =\sum_{\ell=1}^m \widehat{\beta}_{\ell}\cdot \1{\{ \normV{ [\widehat{\Gamma}]^{-1}_{\um}}\leq \gamma \} }\cdot \psi_{\ell}(t), \quad t \in [0,1],
\end{equation}
where the  $\widehat{\beta}_{\ell}$ are the generic elements of the vector of coordinates  obtained by least squares projection and   $\1$ is the indicator function. This new thresholding  step can be seen  as an improvement of the estimator proposed by \cite{RamsayDalzell1991} which was built by projecting the data onto finite dimensional basis of functions. From an inverse problems perspective this approach is similar to the linear Galerkin procedure (\cite{Natterer1977} or \cite{EHN00}) defined as follows, $\beta^{m}\in\Psi_m$ denotes a Galerkin solution of the operator equation $g=\Gamma\beta$ when  
\begin{equation}\label{app:unknown:Galerkin}
\normV{g-\Gamma\beta^m}\leqslant  \normV{g-\Gamma\tilde\beta},\quad\forall \tilde \beta\in\Psi_m.
\end{equation}
Since $\Gamma$ is strictly positive it follows that $\beta^m=[\beta^m]_{\um}^t[\psi]_{\um}(\cdot)$ with $[\beta^m]_{\um}=[\Gamma]_{\um}^{-1}[g]_{\um}$ is the unique Galerkin solution satisfying $[\Gamma(\beta-\beta^m)]_{\um}=0$.    It has the advantage compared to principal components regression that it does not necessitate to estimate the eigenfunctions of the empirical covariance operator.

We will consider a large class of weighted norms to evaluate the asymptotic rates of converge of the thresholded projection estimators.  
For $f \in L^2[0,1],$ we define
\begin{equation}
\normV{f}_\hw^2 =: \sum_{j=1}^\infty \hw_j |\skalarV{f,\psi_j}|^2
\end{equation}
for  some strictly positive sequence of weights  $(\hw_j)_{j\in\N}$. Then, the performance of the estimator $\widehat{\beta}$ of $\beta$  is evaluated according to the risk $\Ex\|\widehat{\beta}-\beta\|_{\hw}^2,$ called $\cW_{\hw}$-risk in the following, which is simply the $L^2[0,1]$-risk when $\hw_j =1$ for all $j \in \N.$  
% {\dr expected $L^2[0,1]$ squared error\note{\dr mean integretad squared error or $L^2[0,1]$-risk}} when $\hw_j =1$ for all $j \in \N.$  
%We consider first this general setup, however in the next section we illustrate this approach by considering different sequences $\hw$. 
This general framework allows us with appropriate choices of the weight sequence $\hw$ to cover the estimation  of derivatives of $\beta$ as well as the optimal estimation with respect to the mean squared prediction error. Indeed,   the prediction  error of a new value of $Y$ 
given any random function $X_{n+1}$ possessing the same distribution as $X$ and being
independent of $X_1,\dotsc,X_n$ can be evaluated as follows (see for example \cite{CardotFerratySarda2003} or \cite{CrambesKneipSarda2007} for similar setups)
\begin{equation*}
\Ex \Bigl[ \left| \int_0^1 \widehat\beta(s) X_{n+1}(s)ds - \int_0^1 \beta(s) X_{n+1}(s)ds \right|^2\, \Bigl\vert\, \widehat\beta\Bigr]=  \langle \Gamma (\widehat\beta-\beta), (\widehat\beta-\beta) \rangle.
\end{equation*}
Consequently, if we suppose, now for sake of simplicity, that the functions $\psi_j$ are also the eigenfunctions $\phi_j$ of operator $\Gamma$ then it is clear that choosing $\hw_j=  \lambda_j$  leads  to  evaluate, according to the $\hw$-norm, the mean squared prediction error of the estimator.

\medskip

The paper is organized a follows. In section 2, we fix notations  and  we first derive consistency of the estimator in the general case under broad moment assumptions and then prove minimax results under some additional regularity assumptions based on a link condition between the operator $\Gamma$ and the basis  $\{\psi_j\}.$ Section 3 is devoted to  the particular case of trigonometric basis and focuses  on finitely and infinitely smoothing operator $\Gamma$ as well as different regularity conditions for the function $\beta.$ We first consider the case of mean squared prediction error and get asymptotic rates of convergence which are comparable to those of \cite{CrambesKneipSarda2007} in the polynomial case. One remarkable result is that for the exponential case, one can attain the parametric rates up to a power of a $\log n$ factor. Rates of convergence for the function itself and its derivatives are also given. They are similar to those obtained by \cite{HallHorowitz2007} in the case of the estimation of the function itself. 
Finally, a brief section 4 presents the concluding remarks and some perspectives. The proofs are gathered in the Appendix.

\section{Asymptotic properties, the general case}\label{sec:gen}

\subsection{Notations and assumptions.} 
We assume from now on that  the regression  function $\beta$ belongs to some ellipsoid $\cW_{\bw}^\br$, $\br>0$, defined as follows
\begin{equation}\label{sob-scale}
 \cW_\bw^\br := \{f\in L^2[0,1]: \sum_{j=1}^\infty \bw_j |\skalarV{f,\psi_j}|^2=:\normV{f}_\bw^2\leq \br\},\end{equation}
where  $\{\psi_j,j\in\N\}$ is as before some orthonormal basis in $L^2[0,1]$ not necessarily corresponding to the eigenfunctions of $\Gamma$, and the sequence of weights $(\bw_j)_{j\in\N}$ is  non-decreasing. 
%\note{\dr We do not need increasing, we should also include the case $\bw\equiv1$, thus non-decreasing.} with $b_1\geqslant 1$.
Here $\cW_\bw^\br$ captures all the prior information (such as the smoothness) about the unknown slope function $\beta$. 

\paragraph{Matrix and operator notations.} Given $m\geqslant 1$, $\Psi_m$ denotes the subspace of $L^2[0,1]$ spanned by the functions $\{\psi_1,\dotsc,\psi_m\}$. $\Pi_m$ and $\Pi_m^\perp$ denote the orthogonal projections on  $\Psi_m$ and its orthogonal complement $\Psi_m^\perp$ respectively. Given an operator (matrix) $K$, $\normV{K}_{\hw}$ denotes its operator $\cW_\hw$-norm, i.e. $\normV{K}_{\hw}:= \sup_{\normV{f}_{\hw}=1} \normV{Kf}_{\hw}$. 
%and {\dr$\tr(K)$ its trace\note{\dr Do we use the trace?}}. 
The inverse operator (matrix) of $K$ is denoted by $K^{-1}$, the adjoint (transposed) operator (matrix) of $K$ by $K^t$. %, and $\normV{K}_{\hw}:= \sup_{\normV{f}_{\hw}=1} \normV{Kf}_{\hw}.$ 
The identity operator (matrix) is denoted by $I$. For  a vector $v$ and a matrix $K$, the upper $m$ subvector and $m\times m$ sub-matrix is denoted by $[v]_{\um}$ and $[K]_{\um}$ and its entries by $v_{i}$ and $K_{i,j}$ respectively. The diagonal   matrix with entries $v$ is denoted by $\Diag(v)$. $[f]$ and $[K]$ denote the (infinite) vector and matrix of the function $f$ and the operator $K$ with the entries $[f]_{i}=\skalarV{f,\psi_i}$ and $[K]_{i,j}=\skalarV{K\psi_j,\psi_i}$ respectively.  Clearly, $[\Pi_m f]_{\um} =[f]_{\um}$ and if we restrict $\Pi_m K \Pi_m$ to an  operator from $\Psi_m$ into itself, then it has the matrix $[K]_{\um}$. Moreover, $\Pi_{m} f = [f]_{\um}^t [\psi]_{\um}(\cdot)$ and $\Pi_{m} K \Pi_{m} f= [f]_{\um}^t [K]_{\um}[\psi]_{\um}(\cdot)$ with $[\psi]_{\um}(\cdot) = (\psi_1(\cdot), \ldots, \psi_{m}(\cdot))^t$. 

Consider the covariance operator $\Gamma$. We assume throughout the paper that $\Gamma$ is strictly positive definite and hence
the matrix $[\Gamma]_{\um}$ is nonsingular for all $m\in\N$, so that $[\Gamma]_{\um}^{-1}$ always exists. Under this assumption the notation $\Gamma_m^{-1}$ is used for the operator from $L^2[0,1]$ into itself, whose matrix in the basis $\{\psi_j\}$  has the entries $([\Gamma]_{\um}^{-1})_{i,j}$ for $1\leqslant i,j\leqslant m$ and zeroes otherwise.

\paragraph{Moment assumptions.}
The results  derived below involve  additional  conditions on the moments of the  random function $X$, which we formalize now.  Here and subsequently, we denote by $\cX$  the set of all centered   random functions $X$ with finite second moment, i.e., $\Ex\normV{X}^2<\infty$, and strictly positive covariance operator.  Given  $X\in\cX$ consider the random vector $[X]_{\um}$, then its entries $[X]_{j}= \skalarV{X,\psi_j}$ have mean zero and  variance  $[\Gamma]_{j,j}=\skalarV{\Gamma \psi_j,\psi_j}$, but they are not uncorrelated. In fact, $[\Gamma]_{\um}$ is the covariance matrix of $[X]_{\um}$. Since $\Gamma$ is strictly positive definite it follows that $[\Gamma]_{\um}$ is non singular. Therefore,  the random vector $[\Gamma]_{\um}^{-1/2} [X]_{\um}$ has mean zero and identity $I_m$ as covariance matrix. Then we denote by $\cX^{k}_{\eta}$, $k\in\N$, $\eta\geqslant1$, the subset of $\cX$  containing only  random functions $X$ with uniformly bounded  $k$-th moment of the corresponding random variables   $[X]_j/[\Gamma]_{j,j}^{1/2},$ $j\in\N$, and $ ([\Gamma]_{\um}^{-1/2} [X]_{\um})_{j}$, $1\leqslant j\leqslant m,$ $m\in\N$, that is
\begin{multline}\label{form:def:cX}
\cX^{k}_{\eta}:=\Bigl\{ X\in\cX\;\text{ such that }\quad\sup_{j\in\N} \Ex\Bigl|{[X]_j}/{[\Gamma]_{j,j}^{1/2}}\Bigr|^k \leqslant \eta\\\hfill\mbox{ and }\sup_{m\in\N} \sup_{1\leqslant j\leqslant m}  \Ex\Bigl|([\Gamma]_{\um}^{-1/2} [X]_{\um})_{j}\Bigr|^k \leqslant \eta\Bigr\}.
\end{multline}
It is worth noting that in case $X\in\cX$ is a Gaussian random function the corresponding random variables   $[X]_j/[\Gamma]_{j,j}^{1/2},$ $j\in\N$ and  $ ([\Gamma]_{\um}^{-1/2} [X]_{\um})_{j}$, $1\leqslant j\leqslant m,$ $m\in\N$, are Gaussian with mean zero and variance one. Hence, for each $k\in\N$ there exists $\eta$ such that any Gaussian random function $X\in \cX$ belongs also to $\cX^{k}_{\eta}$. Furthermore, in what follows, $\cE^k_\eta$ stands for the set of all centered error terms $\epsilon$ with variance one and finite $k$-th moment, i.e., $\Ex|\epsilon|^k\leqslant \eta$.

\subsection{Consistency.}
 The $\cW_{\hw}$-risk of $\widehat{\beta}$  is essentially determined by the deviation of the estimators of $[g]_{\um}$ and $[\Gamma]_{\um}$,  and  by the regularization error due  to the  projection. 
 The next assertion summarizes then minimal conditions to ensure consistency of  $\widehat{\beta}$ proposed in
\eqref{beta:s0}. 

\begin{prop}\label{gen:prop1}\dr Assume an $n$-sample of $(Y,X)$ satisfying \eqref{intro:e1} with $\sigma>0$. Let $\beta\in \cW_\hw$, $X\in\cX^4_{\eta}$ and $\epsilon\in\cE^4_\eta$, $\eta\geqslant 1$.    Consider the estimator $\widehat{\beta}$  with parameter $m:=m(n)$ and threshold  $\gamma:=\gamma(n)$ are chosen  such that 
$\gamma \geqslant2\normV{[{\Gamma}]^{-1}_{\um}}$ and suppose, as  $n\to\infty,$ that $1/m=o(1)$,   $\gamma \,  (m/n)\, \sup_{1\leqslant j\leqslant m}\{ \hw_j\} =o(1)$, $(m^2/n)=o(1)$ and $\gamma^2   \, (m^{3}/n^{1+1/2})=O(1)$. 
If in addition  $\sup_{m\in\N}\normV{{\Gamma}^{-1}_{m} \Pi_m \Gamma \Pi_m^\perp}_\hw<\infty$, then 
 $\Ex\normV{\widehat{\beta}-\beta}^2_\hw=o(1)$ as $n\to\infty$.
\end{prop}

\begin{rem}The last result covers the case $\hw\equiv1$, i.e., the estimator of $\beta$ is consistent without an additional assumption on $\beta$. However, consistency is only obtained under the condition $\sup_{m\in\N}\normV{\Gamma^{-1}_{m} \Pi_m \Gamma \Pi_m^\perp}_\hw<\infty$, which is known to be sufficient to ensure convergence in the $\cW_\hw$-norm as $m\to \infty$ of the Galerkin solution  $\beta^m=[\beta^m]_{\um}^t[\psi]_{\um}(\cdot)$ with $[\beta^m]_{\um}=[\Gamma]_{\um}^{-1}[g]_{\um}$ to the slope parameter  $\beta$. Furthermore, if $\hw$ is increasing, as in case of a Sobolev norm, then $\widehat{\beta}$  is obviously a consistent estimator only if $\beta \in\cW_\hw$. Moreover, in the last assertion we may replace the condition  $\beta\in \cW_\hw$ by the assumption $\beta\in \cW_\bw$ and $(\hw_j/\bw_j)$ is non-increasing. In this situation we have $\cW_\bw\subset \cW_\hw$ and thus the result still holds true. Roughly speaking this corresponds  to the condition that at least  $p\geqslant s$ derivatives exist in case we want to estimate the $s$-th derivative. \hfill$\square$\end{rem}

\paragraph{Link condition.}In the last assertion the choice of the smoothing parameter $m$ and $\gamma,$  i.e. $\gamma \geqslant2\normV{[{\Gamma}]^{-1}_{\um}},$ depends on  the relation between the covariance operator $\Gamma$ associated to the regressor $X$ and the basis $\{\psi_j\}$ used for the projection, which we  formalize next. Consider the sequence $(\normV{\Gamma \psi_j})_{j\geqslant1}$, which is summable and hence converges to zero since $\Gamma$ is nuclear.  In what follows we impose restriction on  the decay of this sequence. Therefore, consider a strictly positive, monotonically decreasing and summable sequence  of weights $\Gw:=(\Gw_j)_{j\in\N}$ with $\Gw_1=1$.  Then for $s\in\R$ denote by $\norm_{\Gw^s}$ the associated  weighted norm given by  $\normV{f}_{\Gw^s}^2:=\sum_{j=1}^\infty \Gw_j^s |\skalarV{f,\psi_j}|^2$. Let $\cN$ be  the set of all self-adjoint nuclear operator defined on $L^2[0,1]$. Then  for $\Gd\geqslant 1$ define the subset $\cN_{\Gw}^\Gd$ of $\cN$  by 
\begin{equation}\label{bm:link}
\cN_{\Gw}^{\Gd}:=\Bigl\{ \Gamma\in\cN:\quad  \normV{f}_{\Gw^2}^2/d^2\leqslant \normV{\Gamma f}^2\leqslant d^2 \normV{f}_{\Gw^2}^2,\quad \forall f \in L^2[0,1]\Bigr\}.
\end{equation}
A similar condition, but in a different context, can be found, for example, in \cite{NairPereverzevTautenhahn05} and \cite{ChenReiss2008}. Note,  for all  $\Gamma\in\cN_{\Gw}^\Gd$ by using the inequality of \cite{Heinz51} it follows that\footnote{We write $a\asymp_\Gd b$ if  $\Gd^{-1}\leqslant b/a\leqslant \Gd$.}  $\normV{\Gamma\psi_j}\asymp_{{\Gd}} \Gw_j$. Hence, the sequence $(\Gw_j)_{j\in\N}$ has to be summable, i.e., $\sum_j\Gw_j<\infty,$ since $\Gamma$ is nuclear.  We  first consider this general class of operator. However, we illustrate condition \eqref{bm:link} in Section %\ref{sec:infin} 
\ref{sec:ex}  by considering the particular cases of a sequence $\Gw$ with  polynomial or exponential decay which are naturally linked to polynomial or exponential decreasing rates for the eigenvalues of $\Gamma$. To be more precise, if the eigenvalue decomposition of $\Gamma\in\cN$ is given by $\{\lambda_j,\psi_j, j\in\N\}$ then $\Gamma\in\cN_\Gw^\Gd$ if and only if $ \lambda_j\asymp_{{\Gd}}\Gw_j$ for all $j\in\N$.   All the results below are derived under the following basic regularity assumption.  

\begin{assumption}\label{ass:reg} Let $\hw:=(\hw_j)_{j\geqslant 1}$, $\bw:=(\bw_j)_{j\geqslant 1}$ and  $\Gw:=(\Gw_j)_{j\geqslant 1}$ be strictly positive sequences of weights  with $\hw_1= 1$,  $\bw_1= 1$ and  $\Gw_1= 1$ such that   $\bw$ and     $(\bw_j/\hw_j)_{j\geqslant 1}$ are non-decreasing and  $\Gw$ and  $(\Gw_j^2/\hw_j)_{j\geqslant 1}$ are non-increasing  with $\Lw:= \sum_j \Gw_j< \infty$. \end{assumption}
Note that under Assumption \ref{ass:reg}, i.e., $(\bw_j/\hw_j)_{j\geqslant 1}$ is non-decreasing,   the ellipsoid  $\cW_\bw^\br$ is a subset of $\cW_\hw^\br$. Roughly speaking, if $\cW_\bw^\br$ describes $p$-times differentiable functions, then the Assumption \ref{ass:reg} ensures that the $\cW_{\hw}$-risk involves maximal $s\leqslant p$ derivatives. On the other hand if the sequence $\hw$ is decreasing, i.e., the $\cW_\hw$-norm is roughly speaking smoothing, the Assumption \ref{ass:reg} excludes  cases in which $\hw$ decreases faster than the sequence $\Gw^2$. 
%{\dr This however, is due to the well-known limitation of the inequality of \cite{Heinz51}.\note{\dr This is wrong! I would like to replace this sentence by the text in green!!!}} 
 However, in case $\hw\equiv\Gw^2$ we show below that the obtainable optimal-rate is parametric, and hence, whenever  $(\hw_j/\Gw_j^2)=o(1)$ it is parametric too.

The next assertion summarizes now minimal conditions to ensure consistency of the estimator $\widehat{\beta}$ given in
\eqref{beta:s0} when the covariance operator satisfies a link condition. 

 \begin{coro}\label{gen:coro1}\dr Assume an $n$-sample of $(Y,X)$ satisfying \eqref{intro:e1} with $\sigma>0$ and associated   covariance operator  $\Gamma \in\cN_\Gw^\Gd$, $\Gd\geqslant1$.  Let $\beta\in \cW_\bw$, $X\in\cX^4_{\eta}$ and $\epsilon\in\cE^4_\eta$, $\eta\geqslant 1$. Consider  the estimator $\widehat{\beta}$  with threshold  $\gamma= 8\Gd^3/\Gw_m $  and parameter $m:=m(n)$  chosen  such that 
 $1/m=o(1)$,   $  (m/n)\, \sup_{1\leqslant j\leqslant m}\{ \hw_j/\Gw_j\} =o(1)$, $(m^2/n)=o(1)$ and $ m^{3}/(\Gw^2_m\,n^{1+1/2})=O(1)$ as  $n\to\infty$. If in addition  Assumption \ref{ass:reg} is satisfied, then 
 $\Ex\normV{\widehat{\beta}-\beta}_\hw^2=o(1)$ as $n\to\infty$.
\end{coro}
It is worth noting that   the link condition  $\Gamma \in\cN_\Gw^\Gd$ used in the last assertion implies   $\sup_{m\in\N}\normV{\Gamma^{-1}_{m} \Pi_m \Gamma \Pi_m^\perp}_\hw<\infty$ and hence ensures automatically the consistency in the $\cW_\hw$-norm of the Galerkin solution $\beta^m$ as $m\to\infty$. However, in order to obtain a rate of convergence it is necessary to impose additional regularity assumption on the slope parameter $\beta$.  First we derive a lower bound for any estimator when these regularity assumptions are formalized by the condition that $\beta$ belongs to the ellipsoid $\cW_\bw^\br$.

%%%%%%%%%%%%%%%%%%%%%%
\subsection{The lower bound.}
It is well-known  that in general the hardest one-dimensional subproblem does not capture the full difficulty in estimating the solution of an inverse problem even in case of a known operator (for details see e.g. the proof in \cite{MairRuymgaart96}).  In other words, there does not exist two sequences of slope functions $\beta_{1,n},\beta_{2,n}\in \cW_\bw^\br$, which are statistically not consistently distinguishable and which satisfy $\normV{\beta_{1,n}-\beta_{2,n}}^2_\hw\geqslant C \dstar$, where $\dstar$ is the optimal rate of convergence. Therefore we need to consider subsets of $\cW_\bw^\br$ with  growing number of elements in order to get the optimal lower bound. More precisely, we obtain the following lower bound by applying Assouad's cube technique (see e.g. \cite{KorostelevTsybakov1993} or \cite{ChenReiss2008}). Moreover,  the following lower bound is obtained under the additional assumption that distribution of the error term $\epsilon$  is Gaussian with mean zero and variance one, i.e., $\epsilon\sim\cN(0,1)$.
\begin{theo}\label{gen:lower:theo}\dr Assume an $n$-sample of $(Y,X)$ satisfying \eqref{intro:e1} with $\sigma>0$ and  associated   covariance operator  $\Gamma \in\cN_\Gw^\Gd$, $\Gd\geqslant1$. Suppose the error term  $\epsilon\sim\cN(0,1)$ is independent of $X$.  Consider $\cW_\bw^\br,$ $\br>0,$  as set of slope functions.  Let $\kstar:=\kstar(n)$ and $\dstar:=\dstar(\kstar)$ for some $\triangle\geqslant 1$ be chosen such that 
\begin{equation}\label{gen:def:m-gam}
\triangle^{-1}\leqslant \frac{b_{\kstar}}{n\,\hw_{\kstar}}\sum_{j=1}^{\kstar}\frac{\hw_j}{\Gw_j}\leqslant \triangle \quad \mbox{ and }\quad 
\dstar:=\hw_{\kstar}/b_{\kstar}.
\end{equation} 
If in addition the Assumption \ref{ass:reg} is satisfied, then for any estimator $\widetilde{\beta}$ of $\beta$ we have
\begin{equation*} \sup_{\beta \in \cW_\br^\br} \left\{ \Ex\normV{\widetilde{\beta}-\beta}^2_\hw\right\}\geqslant \frac{1}{4\triangle}\cdot \min\Bigl\{\frac{\sigma^2}{2\Gd}, \frac{\rho}{\triangle}\Bigr\} \cdot \dstar.
\end{equation*}
\end{theo}
\begin{rem}The normality and independence assumption on the error term  in the last theorem is only used to simplify the calculation of the distance  between distributions corresponding to different slope functions. However, below we show an upper bound for the estimator $\widehat{\beta}$ in case the error term  $\epsilon\in\cE^k_\eta$  and  the regressor $X\in\cX^k_\eta$ for some $k\in \N$ and  $\eta\geqslant 1$ are only uncorrelated, which includes the particular case of an independent Gaussian error considered in Theorem \ref{gen:lower:theo} as long as $\eta$ is sufficiently large.  Therefore, by applying Theorem \ref{gen:lower:theo} an upper bound of order $\dstar$ implies that  this rate is optimal and hence the estimator $\widehat{\beta}$ is minimax-optimal. 
%{\dg \note{\dr I would like to add this sentences}
Note further that if $(\hw_j/\Gw_j)$ is summable then the order $\dstar$ is parametric. This in particular  is the case when $\hw\equiv\Gw^2$ since   $(\Gw_j)$ is summable.
 \hfill$\square$\end{rem}

\begin{rem}\label{gen:rem:regressor}In case the eigenfunctions of the operator $\Gamma$ are known, the obtainable accuracy of any estimator of $\beta$ is essentially determined by the decay of the  eigenvalues $(\lambda_j)_{j\geqslant 1}$ of $\Gamma$. To be more precise, if  for some  sequence of weights $\Gw:=(\Gw_j)_{j\geqslant 1}$ we have 
 \begin{equation}\label{gen:rem:e1}\exists\, d\geqslant1:\qquad  \lambda_j\asymp_d \Gw_j,\qquad j\geqslant1,\end{equation}
then $\Gw$ determines the obtainable rate of convergence (c.f. \cite{JJ2008}).   If $\{\psi_j\}$ are the eigenfunctions of $\Gamma$, i.e., $\lambda_j=\skalarV{\Gamma\psi_j,\psi_j}$, then the condition \eqref{gen:rem:e1} holds if and only if $\Gamma\in\cN_{\Gw}^d$.
In other words,  the condition $\Gamma\in\cN_{\Gw}^d$ specifies in this situation the  decay of the  eigenvalues of $\Gamma$. However, the set  $\cN_{\Gw}$ also contains operators whose eigenfunctions are not given by $\{\psi_j\}$. Then the corresponding eigenvalues may decay far slower than the sequence of weights $\Gw$. Hence, for these operators the obtainable rate of convergence  may be far slower by using the basis $\{\psi_j\}$  in place of their eigenfunctions.\hfill$\square$\end{rem}

\subsection{The upper bound.}
In the following theorem we provide an upper bound for the estimator $\widehat{\beta}$ defined in \eqref{beta:s0}  by assuming  sequences $\bw$, $\hw$ and $\Gw$  with the additional property that 
\begin{equation}\label{gen:upper:varphi:cond}
\frac{\kstar^{2k}}{\dstar n^k} =O(1),\;\; \frac{\kstar}{\dstar n}\,  \sup_{1\leqslant j\leqslant {\kstar}}\Bigl\{\frac{ \hw_j}{\Gw_j}\Bigr\} = O(1)\;\mbox{and}\; \frac{\kstar^{2+k}}{n^{k/2-1}}  = O(1) \mbox{ for  some $k\in\N$  as }n\to\infty, 
\end{equation}
where  $\kstar:=\kstar(n)$ and $\dstar:=\dstar(\kstar)$ are given by \eqref{gen:def:m-gam}. The next theorem states that the rate $\dstar$ of the lower bound given  in Theorem \ref{gen:lower:theo} provides also an upper bound of 
the estimator $\widehat{\beta}$ defined in \eqref{beta:s0}.

\begin{theo}\label{gen:upper:theo}\dr Assume an $n$-sample of $(Y,X)$ satisfying \eqref{intro:e1} with $\sigma>0$ and  associated   covariance operator  $\Gamma \in\cN_\Gw^\Gd$, $\Gd\geqslant1$. Consider $\cW_\bw^\br$, $\br>0$ as set of slope functions and suppose that the sequences $\bw$, $\hw$ and $\Gw$ satisfy the Assumption \ref{ass:reg}. Let $\kstar:=\kstar(n)$ and  $\dstar:=\dstar(n)$ be given by \eqref{gen:def:m-gam} and suppose \eqref{gen:upper:varphi:cond} is satisfied for some $k\geqslant 4$.    Consider the estimator $\widehat{\beta}$  with parameter $m=\kstar$ and  threshold $\gamma= n \, \max(1, 8\, \Gd^3\,\triangle/  \bw_{\kstar})$. If in addition  $X\in \cX_{\eta}^{4k}$ and $\epsilon\in \cE_\eta^{4k}$, $\eta\geqslant 1$, then we have
$$\sup_{\beta\in\cW_\bw^\br} \Ex\normV{\widehat{\beta}-\beta}^2_\hw\leqslant C\,\dstar\, \eta \, \Gd^{16}\,\triangle^2 \{\sigma^{2}+  \rho \Lw \}.$$
where  $C$  is  a positive constant.
\end{theo}
Thus, we have proved that  the rate $\dstar$ is optimal and hence the estimator $\widehat{\beta}$   is  minimax optimal. 

\begin{rem}It is worth noting that as long as the sequence $b$ is increasing the condition on the threshold $\gamma$ given in Theorem \ref{gen:upper:theo} writes $\gamma=n$ for all sufficiently large $n$. Therefore, only the parameter $m$ has to be chosen data-driven   in order to build   an adaptive estimation procedure.  On the other hand, under the assumptions of Theorem \ref{gen:upper:theo}  the parametric rate cannot be obtained. To be more precise, in case that $\sum_j \hw_j/\Gw_j<\infty$, the rate of the  lower bound in Theorem \ref{gen:upper:theo}  is given by $\dstar=1/n$. But in this case the condition ${\kstar}/({\dstar n})\,  \sup_{1\leqslant j\leqslant {\kstar}}\{{ \hw_j}/{\Gw_j}\} = O(1)$ is not satisfied and hence we cannot apply Theorem \ref{gen:upper:theo}. However, we conjecture that the proposed estimator attains also the parametric rate under a  stronger set of assumptions as, for example, used by  \cite{JohannesSchenk2008}  in order to obtain rate optimal estimation of a linear functional of the slope parameter $\beta$.\hfill$\square$\end{rem}

%%%%%%%%%%%%%%%%%%%%%%%%%%%%%%%%%%%%%%%%%%%%%%%%%%%
\section{Mean squared prediction error and derivative estimation}\label{sec:ex}
In this section  we will suppose that the slope function $\beta$ is an element of the Sobolev space of periodic functions $\cW_p$ for some $p> 0$ 
 given by  \begin{equation*}
 \cW_{p}=\Bigl\{f\in H_{s}: f^{(j)}(0)=f^{(j)}(1),\quad j=0,1,\dotsc,p-1\Bigr\},
 \end{equation*}
where  $H_{p}:= \{ f\in L^2[0,1]:  f^{(p-1)}\mbox{ absolutely continuous }, f^{(p)}\in L^2[0,1]\}$  is a Sobolev space (c.f. \cite{Neubauer1988,Neubauer88}, \cite{MairRuymgaart96} or \cite{Tsybakov04}). Let us first remark that if we consider the sequence of weights  $(b_j^p)_{j\in\N}$ given by
\begin{equation}\label{form:def:b}
b_{1}^p=1\quad\mbox{ and }\quad b_{2j}^p= b_{2j+1}^p=j^{2p},\qquad j\in\N,
\end{equation}
and the trigonometric basis
  \begin{equation}\label{model:def:trigon}
\psi_{1}(t)=1, \quad\psi_{2k}(t)=\sqrt{2}\cos(2\pi k t),\quad \psi_{2k+1}(t)=\sqrt{2}\sin(2\pi k t),\quad k=1,2,\dotsc .
\end{equation}
then the Sobolev space  of periodic functions  is equivalently given by $\cW_{\bw^p}$ defined in \eqref{sob-scale}. Therefore, let us denote by $\cW_p^\br:=\cW_{\bw^p}^\br$, $\br>0$, an ellipsoid in the Sobolev space $\cW_p$.

\paragraph{Mean squared prediction error.} We shall first measure the performance of the estimator by  considering the mean prediction error (MPE), i.e., $\Ex\normV{\widehat\beta-\beta}_\Gamma^2$. In this case, if $\Gamma$ satisfies a link condition, that is $\Gamma \in \cN_{\Gw}^{\Gd}$, $\Gd\geqslant 1$, for some weight sequence  $\Gw$ (see definition \ref{bm:link}), then it follows by using the inequality of \cite{Heinz51} that the MPE  is equivalent to the $\cW_\Gw$-risk, that is $\Ex \normV{\widehat\beta-\beta}_\Gw^2$.
To illustrate the previous results we assume in the following the sequence $(\Gw_j)_{m\in\N}$ to be either  polynomially decreasing, i.e., $\Gw_1=1$ and $\Gw_j  = |j|^{-2a}$, $j\geqslant 2$, for some $a>1/2$,  or   exponentially decreasing, i.e.,  $\Gw_1=1$ and $\Gw_j  = \exp(-|j|^{2a})$, $j\geqslant 2$, for some $a>0$. In the polynomial case easy calculus shows  that  a covariance operator $\Gamma\in\cN_\Gw^\Gd$  acts like integrating   $(2a)$-times and hence it is called {\it finitely smoothing} (c.f. \cite{Natterer84}). Furthermore, if the eigenfunctions of $\Gamma$ are $\{\psi_j\}$, then $\Gamma\in\cN_\Gw^\Gd$ holds if and only if the eigenvalues $\lambda_j$  of $\Gamma$  satisfy  $\lambda_j\asymp_d |j|^{-2a}$, which is the case considered, for example, in  \cite{CrambesKneipSarda2007}. 
On the other hand in the exponential case  it can  easily be seen that the link condition $\Gamma\in\cN_\Gw^\Gd$ implies $\cR(\Gamma)\subset \cW_{p}$ for all $p>0$, therefore  the operator $\Gamma$ is called {\it infinitely smoothing} (c.f. \cite{Mair94}).  Moreover, if the eigenfunctions of $\Gamma$ are $\{\psi_j\}$, then  $\Gamma\in\cN_\Gw^\Gd$  holds if and only if the eigenvalues $\lambda_j$  of $\Gamma$  satisfy  $\lambda_j\asymp_d \exp(-j^{2a})$. To the best of our knowledge this case has not been considered yet in the literature. Since in both cases the  basic regularity assumption \ref{ass:reg} is satisfied, the lower bounds presented in the next assertion follow directly from Theorem \ref{gen:lower:theo}. Here and subsequently, we write $a_n\lesssim b_n$  when there exists $C>0$ such that  $a_n\leqslant C\, b_n$  for all sufficiently large $ n\in\N$ and  $a_n\sim b_n$ when $a_n\lesssim b_n$ and $b_n\lesssim a_n$ simultaneously. 

\begin{prop}\label{MPE:lower}\dr \mbox{Under the assumptions of Theorem \ref{gen:lower:theo} we have for any estimator $ \widetilde{\beta}$}\\[-4ex]
\begin{itemize}\item[(i)] in the polynomial case, i.e. $\Gw_1=1$ and $\Gw_j  = |j|^{-2a}$, $j\geqslant 2$, for some $a>1/2$,    that\\[1ex] 
%\begin{equation*}
\hspace*{5ex}$\sup_{\beta \in \cW_p^\rho} \bigl\{ \Ex\normV{\widetilde{\beta}-\beta}_\Gamma^2\bigr\}\gtrsim n^{-(2p+2a)/(2p+2a+1)} $,
%\end{equation*}
\item[(ii)] in the exponential case, i.e. $\Gw_1=1$ and $\Gw_j  = \exp(-|j|^{2a})$, $j\geqslant 2$, for some $a>0$,    that\\[1ex] 
%\begin{equation*}
\hspace*{5ex}$\sup_{\beta \in \cW_p^\rho} \bigl\{ \Ex\normV{\widetilde{\beta}-\beta}^2_\Gamma\bigr\}\gtrsim n^{-1}(\log n)^{1/2a}$.
%\end{equation*}
\end{itemize}
\end{prop}

On the other hand, if the dimension parameter  $m$ and the threshold $\gamma$ in the definition of  the estimator $\widehat{\beta}$ given in \eqref{beta:s0} are chosen appropriately, then, by applying Theorem \ref{gen:upper:theo},  the rates of the lower bound given in the last assertion also provide, up to a constant,  the upper bound of the risk of the  estimator $\widehat{\beta}$, which is summarized in the next proposition. %Note that  the additional condition \eqref{gen:upper:varphi:cond} is satisfied in the exponential case and   for all $k\geqslant 2+8/(2p+2a-1)$ also in the polynomial case. 

\begin{prop}\label{MPE:upper}\dr Under the assumptions of Theorem \ref{gen:lower:theo} consider the estimator $\widehat{\beta}$\\[-4ex]
\begin{itemize}\item[(i)] in the polynomial case, i.e. $\Gw_1=1$ and $\Gw_j  = |j|^{-2a}$, $j\geqslant 2$, for some $a>1/2$,  with $m\sim n^{1/(2p+2a+1)}$ and threshold $\gamma=n$. If in addition  $k\geqslant 2+8/(2p+2a-1),$ then\\[1ex] 
%\begin{equation*}
\hspace*{5ex}$\sup_{\beta \in \cW_p^\rho} \bigl\{ \Ex\normV{\widehat{\beta}-\beta}_\Gamma^2\bigr\}\lesssim n^{-(2p+2a)/(2p+2a+1)}$,
%\end{equation*}
\item[(ii)] in the exponential case, i.e. $\Gw_1=1$ and $\Gw_j  = \exp(-|j|^{2a})$, $j\geqslant 2$, for some $a>0$,  with $m\sim (\log n)^{1/(2a)}$ and threshold $\gamma=n$. Then\\[1ex] 
%\begin{equation*}
\hspace*{5ex}$\sup_{\beta \in \cW_p^\rho} \bigl\{ \Ex\normV{\widehat{\beta}-\beta}^2_\Gamma\bigr\}\lesssim n^{-1}(\log n)^{1/2a}$.
%\end{equation*}
\end{itemize}
 \end{prop}
We have thus proved that these rates are  optimal and the proposed estimator $\widehat{\beta}$ is minimax optimal in both cases. It is worth noting  that replacing the condition $\gamma=n$ by $\gamma=c\,n$ with $c>0$ appropriately chosen,  Proposition  \ref{MPE:upper} remains true when $p=0$, that is to say when $\beta$ is just supposed to be square integrable.%\note{How about this?}
%{\dr \note{Maybe we can write this} Note that Proposition  \ref{MPE:upper} remains true when $p=0,$ that is to say when $\beta$ is just supposed to be a periodic function belonging to $L^2[0,1],$ provided  that $\gamma\sim n.$}

\begin{rem}It is of interest to compare our results with those of \cite{CrambesKneipSarda2007} who measure the performance of their estimator in terms of squared prediction error. In their notations the decay of the eigenvalues of $\Gamma$ is assumed to be of order $ (|j|^{-2q-1})$, i.e., $q=a-1/2$.  Furthermore they  suppose the slope function to be $m$-times continuously differentiable, i.e., $m=p$. By using this parametrization we see that our results in the polynomial case imply the same rate of convergence in probability of the prediction error as it is presented in  \cite{CrambesKneipSarda2007}.  However, from our general results follows a lower and an upper bound of the MPE not only in the polynomial case but also  in the exponential case. 

Furthermore, we shall emphasize the interesting influence of the parameters $p$ and $a$ characterizing the smoothness of $\beta$ and the smoothing properties of $\Gamma$, respectively. As we see from Propositions \ref{MPE:lower} and \ref{MPE:upper},  in the polynomial case  an increasing  value of $p$ leads to a faster optimal rate. In other words, as expected, a smoother  regression function can be faster estimated. 
%{\dr Moreover, if the value of $a$ increases the obtainable optimal rate of convergence decreases. Therefore, the parameter $a$ is often called {\it degree of ill-posedness} (c.f. \cite{Natterer84}).\note{\dr Wrong!! If $a$ increases the optimal rate decreases. I would cut out this two phrases!!!}} 
The situation in the exponential case is extremely different. It seems rather surprising that, contrary  to the polynomial case, in the exponential case the optimal rate of convergence does not depend on the value of $p$, however this dependence is clearly hidden in the constant. Furthermore, the parameter $m$  does not even depend on the value of $p$.  Thereby, the proposed estimator is automatically adaptive, i.e., it does not involve an a-priori knowledge of the degree of smoothness of the slope function $\beta$. However, the choice of the smoothing parameter depends on  the value $a$ specifying the decay of $\{\Gw_j\}$. Note further that in  both cases %{\dr the exponential case\note{\dr I would write " in both cases"}}
 an increasing value of $a$ leads to a faster optimal rate of convergence, i.e., we may call $1/a$  as {\it degree of ill-posedness} (c.f. \cite{Natterer84}). \hfill$\square$\end{rem}
\paragraph{Estimation of the derivatives.} 
Let us consider now the estimation of derivatives of the slope function $\beta$. It is well-known, that for any function $g$ belonging to a  Sobolev-ellipsoid $\cW_{p}^\br$   the Sobolev norm  $\normV{g}_{\bw^s}$ for each $0\leqslant s\leqslant p$  is equivalent to the $L^2$-norm of the $s$-th weak derivative $g^{(s)}$, i.e.,  $\normV{g^{(s)}}$. Thereby, the results in the previous Section %\ref{sec:sob} 
imply again a lower bound as well as an upper bound of the $L^2$-risk for the estimation of the  $s$-th weak derivative of  $\beta$. In the following we consider again the two particular cases of polynomial and exponential decreasing rates for the sequence of weights $(\Gw_j)$. The next assertion summarizes then lower bounds for the $L^2$-risk for the  estimation  of the $s$-th weak derivative of $\beta$ in both cases. 

\begin{prop}\label{MSE:lower}\dr \mbox{Under the assumptions of Theorem \ref{gen:lower:theo} we have for any estimator $ \widetilde{\beta}^{(s)}$}\\[-4ex]
\begin{itemize}\item[(i)] in the polynomial case, i.e. $\Gw_1=1$ and $\Gw_j  = |j|^{-2a}$, $j\geqslant 2$, for some $a>1/2$,    that\\[1ex] 
%\begin{equation*}
\hspace*{5ex}$\sup_{\beta \in \cW_p^\rho} \bigl\{ \Ex\normV{\widetilde{\beta}^{(s)}-\beta^{(s)}}^2\bigr\}\gtrsim n^{-(2p-2s)/(2p+2a+1)} $,
%\end{equation*}
\item[(ii)] in the exponential case, i.e. $\Gw_1=1$ and $\Gw_j  = \exp(-|j|^{2a})$, $j\geqslant 2$, for some $a>0$,    that\\[1ex] 
%\begin{equation*}
\hspace*{5ex}$\sup_{\beta \in \cW_p^\rho} \bigl\{ \Ex\normV{\widetilde{\beta}^{(s)}-\beta^{(s)}}^2\bigr\}\gtrsim (\log n)^{-(p-s)/a}$.
%\end{equation*}
\end{itemize}
\end{prop}

On the other hand considering the estimator $\widehat\beta$ given in \eqref{beta:s0}, we only have to calculate the $s$-th weak derivative of $\beta$. Given the exponential basis, which is linked to the trigonometric basis by the relation $\exp(2 \i \pi k  t) = 2^{-1/2} ( \psi_{2k}(t) + \i \ \psi_{2k+1}(t)),$ for  $k \in \Z$ and $t \in [0,1],$ with $\i^2=-1,$ we recall that    for $0\leqslant s <p$ the $s$-th derivative $\beta^{(s)}$ of $\beta$ in a weak sense  satisfies 
%\note{Attention $\psi_j$ are here complex exponentials!!! They are not defined yet.}
$$
\beta^{(s)}(t) =\sum_{k\in\Z} ( 2\i \pi k)^s  \left( \int_0^1 \beta(u) \exp(-2 \i \pi k u) \ du \right)  \exp(2 \i  \pi k t).
$$

 Given a dimension $m\geqslant 1,$ we denote now by $[\widehat{\Gamma}]_{\um}$ the $(2m+1)\times (2m+1)$ matrix with generic elements $\skalarV{\widehat{\Gamma} \psi_{\ell},\psi_{j}}, -m\leqslant j,\ell\leqslant m$ and by $[\widehat{g}]_{\um}$ the $2m+1$ vector with elements $\skalarV{\widehat{g},\psi_{\ell}}, -m\leqslant \ell\leqslant m.$ Furthermore for integer $s$ define the diagonal matrix $\bigtriangledown^{1/2}_{\um}$ with entries $\bigtriangledown^{1/2}_{j,j}:=(2 \i \pi j)^s$, $-m\leqslant j\leqslant m$. Then we consider the estimator of $\beta^{(s)}$ defined by 
\begin{multline}\label{sob:def:est}
\widehat{\beta}^{(s)}:=[\widehat{\beta}^{(s)}]_{\um}^t [\psi]_{\um}(\cdot)\quad\mbox{ with}\\
[\widehat{\beta}^{(s)}]_{\um}=
\left\{\begin{array}{lcl} 
\bigtriangledown^{s/2}_{\um}[\widehat{\Gamma}]_{\um}^{-1} [\widehat{g}]_{\um}, && \mbox{if $[\widehat{\Gamma}]_{\um}$ is nonsingular}\\&&\qquad\mbox{and }\normV{[\widehat{\Gamma}]^{-1}_{\um}}^2\leqslant \gamma,\\
0,&&\mbox{otherwise}.
\end{array}\right.
\end{multline}
Furthermore,  if the dimension parameter  $m$ and the threshold $\gamma$ in the definition of  the estimator $\widehat{\beta}^{(s)}$ given in \eqref{sob:def:est} are chosen appropriately, then by applying Theorem \ref{gen:upper:theo}  the rates of the lower bound given in the last assertion provide up to a constant again the upper bound of the $L^2$-risk of the  estimator $\widehat{\beta}^{(s)}$, which is summarized in the next proposition. We have thus proved that these rates are  optimal and the proposed estimator $\widehat{\beta}^{(s)}$ is minimax optimal in both cases.

\begin{prop}\label{MSE:upper}\dr Under the assumptions of Theorem \ref{gen:lower:theo} consider the estimator $\widehat{\beta}^{(s)}$\\[-4ex]
\begin{itemize}\item[(i)] in the polynomial case, i.e. $\Gw_1=1$ and $\Gw_j  = |j|^{-2a}$, $j\geqslant 2$, for some $a>1/2$,  with $m\sim n^{1/(2p+2a+1)}$ and threshold $\gamma=n$. If in addition  $k\geqslant 2+8/(2p+2a-1),$ then\\[1ex] 
%\begin{equation*}
\hspace*{5ex}$\sup_{\beta \in \cW_p^\rho} \bigl\{ \Ex\normV{\widehat{\beta}^{(s)}-\beta^{(s)}}^2\bigr\}\lesssim n^{-(2p-2s)/(2p+2a+1)}$,
%\end{equation*}
\item[(ii)] in the exponential case, i.e. $\Gw_1=1$ and $\Gw_j  = \exp(-|j|^{2a})$, $j\geqslant 2$, for some $a>0$,  with $m\sim (\log n)^{1/(2a)}$ and threshold $\gamma=n$. Then\\[1ex] 
%\begin{equation*}
\hspace*{5ex}$\sup_{\beta \in \cW_p^\rho} \bigl\{ \Ex\normV{\widehat{\beta}^{(s)}-\beta^{(s)}}^2\bigr\}\lesssim (\log n)^{-(p-s)/a}$.
%\end{equation*}
\end{itemize}
 \end{prop}

\begin{rem}\label{rem:sob:upper:pol:1}It is worth noting that the $L^2$-risk in estimating the slope function $\beta$ itself, i.e., $s=0$, has been considered in \cite{HallHorowitz2007} only in the polynomial case. In their notations the decrease of the eigenvalues of $\Gamma$ is of order $ (|j|^{-\alpha})$, i.e., $\alpha=2a$.  Furthermore the Fourier coefficients of  the slope function decay at least with rate $j^{-\beta}$, i.e., $\beta=p+1/2$. By using this new parametrization we see that we recover the result of \cite{HallHorowitz2007} in the polynomial case with $s=0$, but without the additional assumption $ \beta > \alpha/2+1$ or $\beta >\alpha-1/2$. 

Furthermore, we shall discuss again  the interesting influence of the parameters $p$ and $a$. As we see from Propositions \ref{MSE:lower} and \ref{MSE:upper}, in both cases an  decreasing of the value of $a$ or an increasing of the value $p$  leads to a faster  optimal rate of convergence.
%\note{\dr changed!!}. 
Hence, in opposite to the MPE  by considering the $L^2$-risk the parameter $a$ describes in both cases the {\it degree of ill-posedness}. Furthermore, the estimation of higher derivatives of the slope function, i.e. by considering a larger value of $s$, is as usual only possible with a slower optimal rate. Finally,  
as for the MPE  in the exponential case the parameter $m$  does not depend on the values of $p$ or $s$, hence the proposed estimator is automatically adaptive. \hfill$\square$\end{rem}

\begin{rem}\label{rem:sob:upper:pol:2}There is an interesting hidden issue in the parametrization we have chosen. Consider a classical indirect regression model with known operator given by $\Gamma$, i.e., $Y=[\Gamma\beta](U)+\epsilon$ where $U$ has a uniform distribution on $[0,1]$ and $\epsilon$ is white noise  (for details see e.g. \cite{MairRuymgaart96}). If in addition the operator $\Gamma$ is finitely smoothing, i.e., $(\Gw_j)$ is polynomially decreasing with $\Gw_j=j^{-2a}$, $j\geqslant 2,$ then given an $n$-sample of $Y$ the optimal rate of convergence of the $\cW_s$-risk of any estimator of $\beta$ is of order $n^{-2(p-s)/[2(p+2a)+1]}$,  since $\cR(\Gamma)=\cW_{2a}$ (c.f. \cite{MairRuymgaart96} or \cite{ChenReiss2008}). However, we have shown that in a functional linear model even with estimated operator the optimal rate is of order $n^{-2(p-s)/[2(p+a)+1]}$. Thus comparing both rates we see that in a functional linear model the covariance operator $\Gamma$ has  the {\it degree of ill-posedness} $a$ while the same operator has, in the indirect regression model, a {\it degree of ill-posedness} $(2a)$. In other words in a functional linear model we do not face the complexity of an inversion of $\Gamma$ but only of its square root $\Gamma^{1/2}$. This, roughly speaking, may be seen as a multiplication of the normal equation $YX=\skalarV{\beta,X}X+X\epsilon$ by the inverse of $\Gamma^{1/2}$. Remarking  that $\Gamma$ is also the covariance operator associated to the error term  $\epsilon X,$ the multiplication by the inverse of $\Gamma^{1/2}$ leads, roughly speaking, to white noise.\hfill$\square$\end{rem}

\section{Concluding remarks and perspectives}\label{sec:concl}
We have proposed in this work a new kind of estimation procedures for the regression function and its derivatives in the functional linear model and proved they can attain optimal rates of convergence.

These estimators depend on two parameters which play the role of smoothing parameters, the dimension $m$ of the projection space and the threshold value $\gamma.$ Building data driven rules that can permit to choose automatically the values of these parameters is certainly a topic that deserves further attention and one promising direction is to adapt the selection technique proposed in  \cite{EfromovichKoltchinskii2001}, \cite{GoldPere2000} and \cite{Tsybakov2000}.

Another point of interest is to extend the thresholding approach in order to consider different thresholding rules for different  coordinates in the considered basis. This could lead for instance  with wavelet basis to  estimators that would adapt to sparseness as well as varying regularity of the regression function.

%\newpage
\appendix
\section{Appendix: Proofs}\label{app:proofs}
\subsection{Proofs of Section \ref{sec:gen}}\label{app:proofs:gen}
We begin by defining and recalling notations to be used in the proofs of this section. Given $m>0$, a Galerkin solution of $g=\Gamma\beta$ is denoted by $\beta^{m}\in\Psi_m$ (see equation \eqref{app:unknown:Galerkin}).
 %\begin{equation}\label{app:unknown:Galerkin}
%\normV{g-\Gamma\beta^m}\leqslant  \normV{g-\Gamma\tilde\beta},\quad\forall \tilde \beta\in\Psi_m.
%\end{equation}
%Since $\Gamma$ is strictly positive it follows that $\beta^m=[\beta^m]_{\um}^t[\psi]_{\um}(\cdot)$ with $[\beta^m]_{\um}=[\Gamma]_{\um}^{-1}[g]_{\um}$ is the unique Galerkin solution satisfying $[\Gamma(\beta-\beta^m)]_{\um}=0$.   
Furthermore, we use the notations 
\begin{multline}\label{app:l:upp:def}%{\beta}^m:= [{\beta}^m]_{\um}^t[\psi]_{\um}(\cdot)\quad \mbox{ with }\quad[{\beta}^m]_{\um}:= [\Gamma]_{\um}^{-1}[g]_{\um},\\
\widetilde{\beta}^m:= [\widetilde{\beta}^m]_{\um}^t[\psi]_{\um}(\cdot)\quad \mbox{ with }\quad[\widetilde{\beta}^m]_{\um}:= [{\beta}^m]_{\um}\1\{\normV{[\widehat{\Gamma}]^{-1}_{\um}}\leqslant \gamma\},\\
[\widehat{\Gamma}]_{\um}=\frac{1}{n}\sum_{i=1}^n[X_i]_{\um}[X_i]_{\um}^t,\quad [\tilde X_i]_{\um}:=[\Gamma]_{\um}^{-1/2} [X_i]_{\um} ,\quad [\tilde\Gamma]_{\um}:=\frac{1}{n}\sum_{i=1}^n[\tilde X_i]_{\um}[\tilde X_i]_{\um}^t,\\
[\Xi_{n}]_{\um}:= [\tilde\Gamma]_{\um} - I_m,\quad [T_{n}]_{\um}:=%[\Pi_m\widehat{\Gamma}\Pi_m^t\beta]_m=
\frac{1}{n}\sum_{i=1}^n \skalarV{X_i,\beta-\beta^m} [X_i]_{\um},\quad [W_{n}]_{\um}:=\frac{\sigma}{n}\sum_{i=1}^n \epsilon_i [X_i]_{\um},
\end{multline} 
where $[\widehat{g}]_{\um}- [\widehat{\Gamma}]_{\um} [\beta^m]_{\um} = [T_{n}]_{\um}+ [W_{n}]_{\um}$ with $\Ex [T_{n}]_{\um} =[\Gamma(\beta-\beta^m)]_{\um}=0$ and  $\Ex [W_{n}]_{\um} =0$,  
 $\Ex[\widehat{\Gamma}]_{\um}=[\Gamma]_{\um}$,  $[\tilde\Gamma]_{\um} = [\Gamma]_{\um}^{-1/2} [\widehat{\Gamma}]_{\um}[\Gamma]_{\um}^{-1/2}$ and hence $\Ex [\Xi_{n}]_{\um}\equiv 0$. Moreover, let us introduce the events 
 \begin{multline}\label{app:l:upp:def:o}
\Omega:=\{ \normV{[\widehat{\Gamma}]^{-1}_{\um}}\leqslant \gamma\},\quad  \Omega_{1/2}:= \{\normV{[\Xi_{n}]_{\um}}\leqslant 1/2\}\\
\Omega^c:=\{ \normV{[\widehat{\Gamma}]^{-1}_{\um}}> \gamma\}\quad\mbox{ and }\quad  \Omega_{1/2}^c=\{\normV{[\Xi_{n}]_{\um}}> 1/2\}.
\end{multline}
Observe that $\Omega_{1/2} \subset\Omega$ in case $\gamma \geqslant2\normV{[{\Gamma}]^{-1}_{\um}}$. Indeed, if $\normV{[\Xi_{n}]_{\um}}\leqslant 1/2$  then the identity $[\widehat{\Gamma}]_{\um}= [{\Gamma}]^{1/2}_{\um}\{I+[\Xi_{n}]_{\um}\}[{\Gamma}]^{1/2}_{\um}$ implies by the usual Neumann series argument that $\normV{[\widehat{\Gamma}]^{-1}_{\um}} \leqslant 2\normV{[{\Gamma}]^{-1}_{\um}}$. Thereby, if $\gamma \geqslant 2\normV{[{\Gamma}]^{-1}_{\um}}$, then we have $\Omega_{1/2} \subset\Omega$. These results will be used below without further reference.

We shall prove in the end of this section the two technical Lemma \ref{pr:upper:l2} and \ref{app:lower:lem1} which are used in the following proofs. 
\paragraph{Proof of the consistency.}
\begin{proof}[\noindent\textcolor{darkred}{\sc Proof of Proposition \ref{gen:prop1}.}] The proof is based on the  decomposition \begin{equation}\label{gen:dec}
 \Ex\normV{\widehat{\beta}-\beta}^2_\hw\leqslant 2\{\Ex\normV{\widehat{\beta}-
 \widetilde{\beta}^m}^2_\hw+\Ex\normV{\widetilde{\beta}^m- {\beta}}^2_\hw\}.
\end{equation}
 Since $\gamma \geqslant2\normV{[{\Gamma}]^{-1}_{\um}}$ it follows that $ \Omega^c\subset \Omega^c_{1/2}$ and hence
\begin{equation}\label{pr:upper:gen:prop1:e2}
\Ex\normV{\widetilde{\beta}^m-\beta}^2_\hw\leqslant 2\{\normV{\beta^m-\beta}^2_\hw+ \normV{\beta^m}^2_\hw\,P(\Omega_{1/2}^c)\}.
\end{equation}
On the other hand  we show  below  for some constant $C>0$  the following  bound
\begin{multline}\label{pr:upper:gen:prop1:e1}
\Ex\normV{\widehat{\beta}- \widetilde{\beta}^m}^2_{\hw}
\leqslant C\cdot \normV{[\Diag(\hw)]^{1/2}_{\um}\,[{\Gamma}]_{\um}^{-1/2}}^2 \, (m/n)\, \eta \, \Bigl\{\sigma^{2}+  \normV{\beta-\beta^m}^{2}\, \Ex\normV{X}^2 \Bigr\} \\
\,\Bigl\{ 1 + \gamma^2   \, m^{2}/n \, \eta^{-1/2}(P(\Omega_{1/2}^c))^{1/2} \normV{[{\Gamma}]_{\um}}^2  \Bigr\},  \end{multline}
where  by applying Markov's inequality    \eqref{pr:upper:l2:e3}   in Lemma \ref{pr:upper:l2} implies $P(\Omega_{1/2}^c)\leqslant C \eta m^2/n$ for some $C>0$. Moreover, $\normV{[{\Gamma}]_{\um}}^2 \leqslant \normV{{\Gamma}}^2$ and  $\normV{[\Diag(\hw)]^{1/2}_{\um}\,[{\Gamma}]_{\um}^{-1/2}}^2\leqslant \gamma \sup_{1\leqslant j\leqslant m} \{\hw_j\} $ since $\gamma \geqslant2\normV{[{\Gamma}]^{-1/2}_{\um}}^2$, which by combination of \eqref{pr:upper:gen:prop1:e2} and \eqref{pr:upper:gen:prop1:e1} leads to the estimate
\begin{multline}\label{pr:upper:gen:prop1:e3}
\Ex\normV{\widehat{\beta}- \beta}^2_{\hw}
\leqslant C\,\Bigl\{ \normV{\beta^m-\beta}^2_\hw+ \normV{\beta^m}^2_\hw\,(m^2/n)\,\eta \\+   \gamma \sup_{1\leqslant j\leqslant m}\{ \hw_j\} \, (m/n)\, \eta \, \{\sigma^{2}+  \normV{\beta-\beta^m}^{2}\, \Ex\normV{X}^2 \} 
\,\{ 1 + \gamma^2   \, (m^{3}/n^{1+1/2}) \,  \normV{{\Gamma}}^2 \} \end{multline}
for some $C>0$. Furthermore, for each $\beta\in \cW_\hw$, we have  $\normV{\beta-\beta^m}_\hw=o(1)$ as $m\to\infty$, which can be realized as follows. Since  $\normV{\Pi_m^\perp  \beta}=o(1)$ and $\normV{\Pi_m^\perp  \beta}_\hw=o(1)$ as $m\to\infty$ by using Lebesgue's dominated convergence theorem, the assertion follows from the identity $[\Pi_m \beta-\beta^m]_{\um} = -[\Gamma]_{\um}^{-1}[\Gamma \Pi_m^\perp \beta]_{\um}$ by using  that $\normV{\Pi_m \beta -\beta^m}_\hw\leqslant\normV{\Pi_m^\perp \beta   }_\hw \sup_m \normV{{\Gamma}^{-1}_{m} \Pi_m \Gamma \Pi_m^\perp}_\hw= O(\normV{\Pi_m^\perp \beta   }_\hw)$. Consequently, the
conditions on  $m$ and $\gamma$  ensure the convergence to zero as $n\to\infty$ of  the  bound given  in \eqref{pr:upper:gen:prop1:e3}, which proves the result.

Proof of \eqref{pr:upper:gen:prop1:e1}. From the 
 identity $[\widehat{g}]_{\um}-[\widehat{\Gamma}]_{\um}[\beta^m]_{\um} =[T_{n}]_{\um}+ [W_{n}]_{\um}$ it follows that
\begin{equation*}
\Ex\normV{\widehat{\beta}- \widetilde{\beta}^m}^2_\hw=
\Ex\normV{[\Diag(\hw)]^{1/2}_{\um}\,\{[{\Gamma}]_{\um}^{-1}+[\widehat{\Gamma}]_{\um}^{-1}([{\Gamma}]_{\um} - [\widehat{\Gamma}]_{\um})[{\Gamma}]_{\um}^{-1}\}\,\{[T_{n}]_{\um}+ [W_{n}]_{\um}\}}^{2}\1_\Omega.
\end{equation*}
Since $ 2\normV{[{\Gamma}]^{-1}_{\um}}\leqslant \gamma$ we have $\Omega_{1/2}\subset\Omega$, and  hence by using $ \normV{[\widehat{\Gamma}]_{\um}^{-1}}^2\1_{\Omega} \leqslant\gamma^2$ we obtain
\begin{multline*}
\Ex\normV{\widehat{\beta} -\widetilde{\beta}^m}_\hw^2\leqslant
3 \normV{[\Diag(\hw)]^{1/2}_{\um}\,[{\Gamma}]_{\um}^{-1/2}}^2 \, \Bigl\{ \Ex\normV{[{\Gamma}]_{\um}^{-1/2}\{[T_{n}]_{\um}+ [W_{n}]_{\um}\}}^2 \\
\hfill+    \gamma^2 \, \normV{[{\Gamma}]_{\um}}^2 \,  (\Ex\normV{[\Xi_{n}]_{\um}}^8)^{1/4}(\Ex\normV{[{\Gamma}]_{\um}^{-1/2}\{[T_{n}]_{\um}+ [W_{n}]_{\um}\}}^8)^{1/4}(P(\Omega_{1/2}^c))^{1/2}\\
+ \Ex\normV{\{I+[\Xi_{n}]_{\um}\}^{-1}}^2\normV{[\Xi_{n}]_{\um}}^2\normV{[{\Gamma}]_{\um}^{-1/2}\,\{[T_{n}]_{\um}+ [W_{n}]_{\um}\}}^{2}\1_{\Omega_{1/2}}\Bigr\}.
\end{multline*}
From \eqref{pr:upper:l2:e1}-\eqref{pr:upper:l2:e3}   in Lemma \ref{pr:upper:l2} together with $\normV{\{I+[\Xi_{n}]_{\um}\}^{-1}}\normV{[\Xi_{n}]_{\um}}\1_{\Omega_{1/2}}\leqslant 1$   follows then \eqref{pr:upper:gen:prop1:e1}, which completes the proof.\end{proof}
% %%%%%%%%%%%%%%%%%%%%%%%%%%%%%%%%%%%%%%%%%%

\begin{proof}[\noindent\textcolor{darkred}{\sc Proof of Corollary \ref{gen:coro1}.}] The link condition $\Gamma\in\cN^\Gd_\Gw$ implies $2\normV{[\Gamma]^{-1}_{\um}}\leqslant 8\Gd^3/\Gw_m= \gamma$, $\normV{[\Diag(\hw)]^{1/2}_{\um}\,[{\Gamma}]_{\um}^{-1/2}}^2\leqslant 4\Gd^3 \sup_{1\leqslant j\leqslant m} \{\hw_j/\Gw_j\} $ and 
$\normV{[{\Gamma}]_{\um}}^2\leqslant \Gd^2 $ by using the estimates \eqref{pr:upper:l3:e1-1}, \eqref{pr:upper:l3:e1-2}  and \eqref{pr:upper:l3:e1-3}  in Lemma \ref{pr:upper:l3}, respectively. Therefore, by combination of \eqref{pr:upper:gen:prop1:e2} and \eqref{pr:upper:gen:prop1:e1} in the proof of Proposition \ref{gen:prop1} we obtain
\begin{multline}\label{pr:upper:gen:coro1:e1}
\Ex\normV{\widehat{\beta}- \beta}^2_{\hw}
\leqslant C\,\Bigl\{ \normV{\beta^m-\beta}^2_\hw+ \normV{\beta^m}^2_\hw\,(m^2/n)\,\eta +  \Gd^3 \sup_{1\leqslant j\leqslant m}\{ \hw_j/\Gw_j\} \, (m/n)\\\, \eta \, \{\sigma^{2}+  \normV{\beta-\beta^m}^{2}\, \Ex\normV{X}^2 \} 
\,\{ 1 +  m^{3}/(n^{1+1/2}\Gw_m^2) \,   \Gd^8 \} \Bigr\}\end{multline}
for some $C>0$. By using  the identity $[\Pi_m \beta-\beta^m]_{\um} = -[\Gamma]_{\um}^{-1}[\Gamma \Pi_m^\perp \beta]_{\um}$ and the estimate \eqref{pr:upper:l3:e2:1} in the proof of Lemma \ref{pr:upper:l3} with $\bw\equiv \hw$  the link condition $\Gamma\in\cN^\Gd_\Gw$ implies further that $\normV{\Gamma^{-1}_{m} \Pi_m \Gamma \Pi_m^\perp}_\hw^2=\sup_{\normV{\beta}_\hw=1}\normV{\Pi_m \beta-\beta^m}^2_\hw\leqslant  2(1+\Gd^2)$ for all $m\in\N$. 
Therefore  we have  $\normV{\beta-\beta^m}_\hw=o(1)$ as $m\to\infty$ for each $\beta\in \cW_\hw$. Consequently, the
conditions on  $m$ and $\gamma$  ensure the convergence to zero as $n\to\infty$ of  the  bound given  in \eqref{pr:upper:gen:coro1:e1}, which proves the result.\end{proof}

\paragraph{Proof of the lower bound.}
\begin{proof}[\noindent\textcolor{darkred}{\sc Proof of Theorem \ref{gen:lower:theo}.}]  Let $X_i$, $i\in\N$, be i.i.d. copies of $X$ with associated covariance operator $\Gamma$ belonging to $\cN_\Gw^\Gd$. Then for each $j$, $[X_i]_j$ is centered and has variance $\Ex [X]_j^2=\skalarV{\Gamma\psi_j,\psi_j}\leqslant \Gw_j\Gd $. This result will be used below without further reference.  Consider independent error terms $\epsilon_i\sim \cN(0,1)$, $i\in\N$, which are  independent of the random functions $\{X_i\}$.  Let  $\theta\in\{-1,1\}^{\kstar}$, where $\kstar:=\kstar(n)\in\N$ satisfies \eqref{gen:def:m-gam} for some $\triangle\geqslant 1$. Define a  $\kstar$-vector $u$ of coefficients $u_{j}$  satisfying \eqref{app:lower:lem1:u} in Lemma \ref{app:lower:lem1}. For each $\theta$ we consider a slope function $\beta^\theta:=\sum_{j=1}^ {\kstar}\theta_{j} u_{j}\psi_j\in\cW_p^\rho$ by using \eqref{app:lower:lem1:e1}  in Lemma \ref{app:lower:lem1}. Consequently, for each $\theta$ the random variables $(Y_i,X_i)$ with $Y_i:=\int_0^1\beta^{\theta}(s)X_i(s)ds+\sigma\epsilon_i$, $i=1,\dotsc,n$, form a sample of the model \eqref{intro:e1} and we denote its joint distribution by  $P_{\theta}$. Furthermore, for $j=1,\dotsc,\kstar$ and each $\theta$ we introduce $\theta^{(j)}$ by $\theta^{(j)}_{l}=\theta_{l}$ for $j\ne l$ and $\theta^{(j)}_{j}=-\theta_{j}$.  As in case of  $P_\theta$ the conditional distribution of $Y_i$ given $X_i$  is Gaussian with mean  $\sum_{j=1}^{\kstar} \theta_{j} u_{j} [X_i]_j$  and variance $\sigma^2$ it is easily seen that   the log-likelihood of $P_{\theta^{(j)}}$ w.r.t. $P_{\theta}$ is given by 
\begin{equation*}
\log\Bigl(\frac{dP_{\theta^{(j)}}}{dP_{\theta}}\Bigr)=-\frac{1}{\sigma^2}   \sum_{i=1}^n 
\Bigl\{Y_i - \sum_{l=1}^{\kstar} \theta_{l} {u}_{l}[X_i]_l \Bigr\} \theta_{j} {u}_{j}[X_i]_j - \frac{2}{\sigma^2} \sum_{i=1}^n {u}_{j}^2[X_i]_j^2
\end{equation*}
and its expectation w.r.t. $P_{\theta}$ satisfies 
%\begin{equation*}
$\Ex_{P_{\theta}}[\log(dP_{\theta^{(j)}}/dP_{\theta})]\geqslant -2n\Gd\, u^2_{j}\,  \Gw_j/ \sigma^2$. In terms of  Kullback-Leibler divergence this means $KL(P_{\theta^{(j)}},P_{\theta})\leqslant 2n\Gd\, u^2_{j}\,  \Gw_j/ \sigma^2$. Since the
 Hellinger distance $H(P_{\theta^{(j)}},P_{\theta})$ satisfies $H^2(P_{\theta^{(j)}},P_{\theta}) \leqslant KL(P_{\theta^{(j)}},P_{\theta})$  it follows from \eqref{app:lower:lem1:e1} in Lemma \ref{app:lower:lem1} that 
\begin{equation}\label{pr:lower:e3}
H^2(P_{\theta^{(j)}},P_{\theta}) \leqslant \frac{2 n \Gd }{\sigma^2}\cdot u^2_{j}\cdot \Gw_j\leqslant 1,\quad j=1,\dotsc,\kstar.
\end{equation} 
Consider  the  Hellinger affinity $\rho(P_{\theta^{(j)}},P_{\theta})= \int \sqrt{dP_{\theta^{(j)}}dP_{\theta}}$, then we obtain for any estimator $\widetilde\beta$ of $\beta$ that
\begin{align}\nonumber
\rho(P_{\theta^{(j)}},P_{\theta})&\leqslant \int \frac{|\skalarV{\widetilde{\beta}-\beta^{\theta^{(j)}},\psi_j}|}{|\skalarV{\beta^\theta-\beta^{\theta^{(j)}},\psi_j}|} \sqrt{dP_{\theta^{(j)}}dP_{\theta}} + \int \frac{|\skalarV{\widetilde{\beta}-\beta^\theta,\psi_j}|}{|\skalarV{\beta^\theta-\beta^{\theta^{(j)}},\psi_j}|} \sqrt{ dP_{\theta^{(j)}}dP_{\theta}}\\\label{pr:lower:e4}
&\leqslant \Bigl( \int  \frac{|\skalarV{\widetilde{\beta}_s-\beta^{\theta^{(j)}},\psi_j}|^2}{|\skalarV{\beta^\theta-\beta^{\theta^{(j)}},\psi_j}|^2} dP_{\theta^{(j)}}\Bigr)^{1/2} +   \Bigl( \int  \frac{|\skalarV{\widetilde{\beta}-\beta^\theta,\psi_j}|^2}{|\skalarV{\beta^\theta-\beta^{\theta^{(j)}},\psi_j}|^2} dP_{\theta}\Bigr)^{1/2}.
\end{align}
Due to the identity $\rho(P_{\theta^{(j)}},P_{\theta})=1-\frac{1}{2}H^2(P_{\theta^{(j)}},P_{\theta})$   combining  \eqref{pr:lower:e3} with 
 \eqref{pr:lower:e4} yields
\begin{equation*}
\Bigl\{\Ex_{{\theta^{(j)}}}|\skalarV{\widetilde{\beta}-\beta^{\theta^{(j)}},\psi_j}|^2+ \Ex_{{\theta}}|\skalarV{\widetilde{\beta}-\beta^\theta,\psi_j}|^2\Bigr\}\geqslant\frac{1}{2} u_j^2,\quad j=1,\dotsc, \kstar.
 \end{equation*}
From this  we conclude for each estimator $\widetilde\beta$ that
\begin{align*}
\sup_{\beta \in \cW_\bw^\rho} &\Ex\normV{\widetilde\beta -\beta}^2_\hw \geqslant \sup_{\theta\in \{-1,1\}^{\kstar}} \Ex_\theta\normV{\widetilde\beta -\beta^\theta}_\hw^2\\
&\geqslant \frac{1}{2^{{\kstar}}}\sum_{\theta\in \{-1,1\}^{\kstar}}\sum_{j=1}^{\kstar}\hw_j\Ex_{{\theta}}|\skalarV{\widetilde{\beta}-\beta^{\theta},\psi_j}|^2\\
&= \frac{1}{2^{{\kstar}}}\sum_{\theta\in \{-1,1\}^{\kstar}}\sum_{j=1}^{\kstar}\hw_j\frac{1}{2}\Bigl\{\Ex_{{\theta}}|\skalarV{\widetilde{\beta}-\beta^\theta,\psi_j}|^2+\Ex_{{\theta^{(j)}}}|\skalarV{\widetilde{\beta}-\beta^{\theta^{(j)}},\psi_j}|^2 \Bigr\}\\
&\geqslant\frac{1}{4} \sum_{j=1}^{\kstar}  u_j^2\cdot \hw_j \geqslant \frac{1}{4}\cdot \min\Bigl\{\frac{\sigma^2}{2d}, \frac{\br}{\triangle}\Bigr\} \cdot \frac{\dstar}{\triangle},
\end{align*}
where the last inequality follows from  \eqref{app:lower:lem1:e1} in Lemma \ref{app:lower:lem1}  which completes the proof.\end{proof}

\paragraph{Proof of the upper bound.}
\begin{proof}[\noindent\textcolor{darkred}{\sc Proof of Theorem \ref{gen:upper:theo}.}] Our proof starts with the observation that the link condition $\Gamma\in\cN^\Gd_\Gw$ implies $2\normV{[\Gamma]^{-1}_{\um}}\leqslant 8\Gd^3/\Gw_m$, $\normV{[\Diag(\hw)]^{1/2}_{\um}\,[{\Gamma}]_{\um}^{-1/2}}^2\leqslant 4\Gd^3 \sup_{1\leqslant j\leqslant m} \{\hw_j/\Gw_j\} $ and 
$\normV{[{\Gamma}]_{\um}}^2\leqslant \Gd^2 $ by using the estimates \eqref{pr:upper:l3:e1-1}, \eqref{pr:upper:l3:e1-2}  and \eqref{pr:upper:l3:e1-3}  in Lemma \ref{pr:upper:l3}, respectively. Moreover, for all $X\in \cX^{4k}_\eta$  by applying Markov's inequality    \eqref{pr:upper:l2:e3}   in Lemma \ref{pr:upper:l2} we have  $P(\Omega_{1/2}^c)\leqslant C \eta m^{2k}/n^k$ for some $C>0$. Furthermore, by using   the definition of $\kstar$
the condition $m=\kstar$ implies $1/\Gw_{\kstar}\leqslant n\triangle/\bw_{\kstar}$ and hence  $\gamma = n\max(1, 8\, \Gd^3\,\triangle/  b_{\kstar})\geqslant 2\normV{[\Gamma]^{-1}_{\ukstar}}$. Therefore, from  \eqref{pr:upper:gen:prop1:e2} and \eqref{pr:upper:gen:prop1:e1} in the proof of Proposition \ref{gen:prop1} follows \begin{multline*}
\Ex\normV{\widehat{\beta}- \beta}^2_{\hw}
\leqslant C\,\Bigl\{ \normV{\beta^{\kstar}-\beta}^2_\hw+ \normV{\beta^{\kstar}}^2_\hw\,(\kstar^{2k}/n^k)\,\eta +  \Gd^3 \sup_{1\leqslant j\leqslant {\kstar}}\{ \hw_j/\Gw_j\} \, ({\kstar}/n)\\\, \eta \, \{\sigma^{2}+  \normV{\beta-\beta^{\kstar}}^{2}\, \Ex\normV{X}^2 \} 
\,\{ 1 +   \kstar^{2+k}/(n^{k/2-1}) \,  \Gd^8\,\triangle^2  \} \Bigr\}\end{multline*}
for some $C>0$. Consequently, the definition of $\dstar$ by using \eqref{pr:upper:l3:e2} in Lemma \ref{pr:upper:l3}, i.e.,  $\normV{\beta-\beta^{\kstar}}^{2}_\hw\leqslant 10 \Gd^4 \,\br \dstar,$ and $\Ex\normV{X}^2\leqslant \Gd \Lw$,  implies
\begin{multline*}
\Ex\normV{\widehat{\beta}- \beta}^2_{\hw}
\leqslant C\,\dstar\, \eta \, \Gd^{16}\,\triangle^2 \{\sigma^{2}+  \rho \Lw \}\,\\ \Bigl\{ 1+ \kstar^{2k}/(\dstar n^k)+   {\kstar}/(\dstar n)\,  \sup_{1\leqslant j\leqslant {\kstar}}\{ \hw_j/\Gw_j\} \,\Bigl\}\,\Bigl\{ 1 +   \kstar^{2+k}/(n^{k/2-1})   \Bigr\}\end{multline*}
Thereby, the result follows from  the condition \eqref{gen:upper:varphi:cond} which ensures that the factors in braces are  bounded as $n\to\infty$, which completes the proof. \end{proof}
% %%%%%%%%%%%%%%%%%%%%%%%%%%%%%%%%%%%%%%%%%%

\paragraph{Technical assertions.}\hfill\\[1ex]
The following two lemma gather technical results used in the proof of Proposition \ref{gen:prop1}, Theorem \ref{gen:lower:theo} and Theorem \ref{gen:upper:theo}.
\begin{lem}\label{pr:upper:l2}
Suppose $X\in\cX_{\eta}^{4k}$ and $\epsilon\in\cE^{4k}_\eta$, $k\in\N$. Then for some constant $C>0$ only depending on $k$ we have
\begin{gather}\label{pr:upper:l2:e1}
\Ex\normV{[\Gamma]_{\um}^{-1/2} W_{n,m}}^{2k}\leq C\cdot \frac{m^{k}}{n^k}\cdot \sigma^{2k}\cdot\eta,\\ \label{pr:upper:l2:e2}
\Ex\normV{[\Gamma]_{\um}^{-1/2} T_{n,m}}^{2k}\leq C\cdot \frac{m^{k}}{n^k}\cdot \normV{\beta-\beta^m}^{2k}\cdot (\Ex\normV{X}^2)^k\cdot\eta,\\ \label{pr:upper:l2:e3}
\Ex\normV{ \Xi_{n,m}}^{2k} \leq C\cdot\eta\cdot \frac{m^{2k}}{n^k},\\ \label{pr:upper:l2:e4}
\Ex\normV{\{[\Gamma]_{\um} - [\widehat{\Gamma}]_{\um}\}[\Gamma]_{\um}^{-1/2}}^{2k} \leq C\cdot\eta\cdot \frac{m^{2k}}{n^k}\cdot (\Ex\normV{X}^2)^k
\end{gather}
\end{lem}
\begin{proof}[\textcolor{darkred}{\sc Proof.}]Let $\tilde W := [\Gamma]_{\um}^{-1/2} W_{n,m}$, then $\Ex\normV{[\Gamma]_{\um}^{-1/2} W_{n,m}}^{2k}\leq m^{k-1} \sum_{j=1}^m\Ex \tilde W^{2k}_j$, where $\tilde W_j = (1/n)\sum_{i=1}^n \sigma\epsilon_i[\tilde X_i]_{j}$. The random variables $(\epsilon_i[\tilde X_i]_{j})$,  $1\leq i\leq n,$ are independent and identically distributed (i.i.d.) with mean zero. Since $X\in\cX_{\eta}^{4k}$ and $\epsilon\in\cE^{4k}_\eta$, \eqref{pr:upper:l2:e1}  follows  from Theorem 2.10 in \cite{Petrov1995}, that is, $\Ex \tilde W_{j}^{2k} \leq C n^{-k}\sigma^{2k} \Ex |\epsilon[\tilde X]_{j}|^{2k}\leq C n^{-k}\sigma^{2k}\eta$ for some constant $C>0$ only depending on $k$.

Proof of \eqref{pr:upper:l2:e2}. Due to  $ \Ex \skalarV{\beta-\beta_m,X}[X]_{\um}= [\Gamma(\beta - \beta^m)]_{\um}=0$, i.e., the random variables $(\skalarV{\beta-\beta^m,X_i}[X_i]_{\um})$,  $1\leqslant i\leqslant n,$ are i.i.d. with mean zero. Furthermore, we claim that  $X\in\cX_{\eta}^{4k}$  implies $\Ex |\skalarV{\beta-\beta^m,X}  [\tilde X]_{j}|^{2k} \leq C\cdot\eta\cdot \normV{\beta-\beta^m}^{2k} (\Ex\normV{X}^2)^k$,  for each $j\in\N$. Then the estimate \eqref{pr:upper:l2:e2} follows  in analogy to \eqref{pr:upper:l2:e1}. Indeed, we have  $\Ex|[\tilde X]_{j}|^{4k}\leqslant \eta$ and 
\begin{align*}
\Ex|\skalarV{\beta-\beta^m,X}|^{4k} &\leqslant \normV{\beta-\beta^m}^{4k} \sum_{j_1}[\Gamma]_{j_1,j_1}\dots\sum_{j_{2k}}[\Gamma]_{j_{2k},j_{2k}} \Ex\prod_{l=1}^{2k} |[X]_{j_l}/[\Gamma]_{j_{l},j_{l}}^{1/2}|^2\\&\leqslant \normV{\beta-\beta^m}^{4k} \,  (\Ex\normV{X}^2)^{2k}\, \eta,
\end{align*}
which  imply together the assertion by using the Cauchy-Schwarz inequality.

Proof of \eqref{pr:upper:l2:e3}. From the identity $(\Xi_{n,m})_{j,l}= (1/n)\sum_{i=1}^n \{[\tilde X_i]_{j}[\tilde X_i]_{l}-\delta_{jl}\}$ with $\delta_{jl}= 1$ if $j=l$ and zero otherwise, we conclude $\Ex (\Xi_{n,m})_{j,l}^{2k}\leq C' n^{-k} \Ex |[\tilde X]_{j}[\tilde X]_{l}-\delta_{jl}|^{2k}$.  Thus $X\in\cX_{\eta}^{4k}$ implies  $\Ex\normV{\Xi_{n,m}}^{2k}\leq m^{2(k-1)} \sum_{j,l} \Ex (\Xi_{n,m})_{j,l}^{2k}\leq C m^{2k} n^{-k}\eta $.

The estimate \eqref{pr:upper:l2:e4} follows by using  the identity  $\{[\Gamma]_{\um} - [\widehat{\Gamma}]_{\um}\}[\Gamma]_{\um}^{-1/2}=[\Gamma]_{\um}^{1/2}\Xi_{n,m}$ from \eqref{pr:upper:l2:e3},  which completes the proof.\end{proof}

%%%%%%%%%%%%%%%%%%%%%%%%%%%%%%%%%%%%%%%%%
%%%%%%%%%%%%%%%%%%%%%%%%%%%%%%%%%%%%%%%%%%%
%%%%%%%%%%%%%%%%%%%%%%%%%%%%%%%%%%%%%%%%%%%
\begin{lem}\label{app:lower:lem1} Let $\kstar\in\N$ and $\dstar$ be chosen such that  \eqref{gen:def:m-gam} is satisfied for some $\triangle\geqslant 1$.  Consider a (infinite) vector $u$ with components $u_j$ satisfying  \begin{equation}\label{app:lower:lem1:u}u_j^2= \frac{\zeta}{n\cdot \Gw_j},\quad j\in\N, \quad \text{ with }\quad \zeta:=\min \left\{ \sigma^2/(2\Gd),  \br/\triangle\right\},\end{equation}
then under Assumption \ref{ass:reg} we have for all $j\in\N$
\begin{gather}\label{app:lower:lem1:e1}
\frac{2n \Gd }{\sigma^2}\,u^2_j \,\Gw_j\leqslant 1,\;
\sum_{j=1}^{\kstar}\,u^2_j \,  \bw_j\leqslant \br,\; \mbox{and}\;
 \sum_{j=1}^{\kstar}\,u^2_j\,  \hw_{j}\geqslant \min \left\{ \frac{\sigma^2}{2\Gd},  \frac{\br}{\triangle}\right\} \,\frac{\dstar}{\triangle}.
\end{gather}
\end{lem}
\begin{proof}[\textcolor{darkred}{\sc Proof.}] The first inequality in \eqref{app:lower:lem1:e1}  follows trivially by using the definition of $\zeta$, while   the definition of $\kstar$ given in \eqref{gen:def:m-gam} together with Assumption \ref{ass:reg}, i.e., $(\bw_{j}/\hw_{j})$ is non-decreasing,
implies the second, i.e.,  $\sum_{j=1}^{\kstar}\,u^2_{j}\, \bw_j \leqslant 
\zeta\, \bw_{\kstar}/\hw_{\kstar}\,\sum_{j=1}^{\kstar}\hw_j/(n \Gw_j) \leqslant\zeta\,\triangle\leqslant \br$. To deduce the third estimate  from the definition of $\kstar$ and $\dstar$  observe that  $ \sum_{j=1}^{\kstar}u^2_j \hw_j =  \dstar \,\zeta \, \bw_{\kstar}/\hw_{\kstar}\ \sum_{j=1}^{\kstar}\hw_j/{(n\, \Gw_j)} \geqslant \dstar \,\zeta/\triangle$, which proves the lemma.\end{proof}
%%%%%%%%%%%%%%%%%%%%%%%%%%%%%%%%%%%%%%%%%%%

\begin{lem}\label{pr:upper:l3}Suppose the sequences $\bw$, $\hw$ and $\Gw$ satisfy Assumption \ref{ass:reg}. Let $\Gamma\in\cN_{\Gw}^\Gd$. Then 
\begin{gather}\label{pr:upper:l3:e1-1}
\sup_{m\in\N}\Bigl\{ \Gw_m \normV{[\Gamma]_{\um}^{-1/2}}^2\Bigr\}\leqslant \{2\Gd^2(2\Gd^4+3)\}^{1/2}\leqslant 4d^3,\\
\label{pr:upper:l3:e1-2}
\sup_{m\in\N}\Bigl\{  \normV{ [\Diag(\Gw)]^{1/2}_{\um} [\Gamma]_{\um}^{-1/2}}^2\Bigr\}\leqslant \{2\Gd^2(2\Gd^4+3)\}^{1/2}\leqslant 4d^3,\\
\label{pr:upper:l3:e1-3}
\sup_{m\in\N}\Bigl\{  \normV{ [\Diag(\Gw)]^{-1/2}_{\um} [\Gamma]_{\um}^{1/2}}^2\Bigr\}\leqslant \Gd.
\end{gather}
If  in addition $\beta^m$  denotes a Galerkin solution of $g=\Gamma\beta$ with $\beta\in \cW_\bw^\br$, then
\begin{gather} \label{pr:upper:l3:e2}
\sup_{m\in\N}\Bigl\{\bw_m/\hw_m\,\normV{\beta-\beta^m}_\hw^2\Bigr\}\leqslant 2(2d^4+3) \, \br\leqslant 10 d^4.
\end{gather}
\end{lem}
\begin{proof}[\noindent\textcolor{darkred}{\sc Proof.}] We start our proof with the observation that   the link condition $\Gamma\in\cN_{\Gw}^\Gd$ implies that $\Gamma$ is strictly positive and that for all $|s|\leqslant 1$  by using the inequality of \cite{Heinz51}
\begin{equation}\label{pr:upper:l3:Heinz}
\Gd^{-2|s|} \normV{f}_{\Gw^{2s}}^2 \leqslant \normV{\Gamma^sf}^2\leqslant \Gd^{2|s|} \normV{f}_{\Gw^{2s}}^2 .
\end{equation}
Consider $g\in\Psi_m$. Then \eqref{pr:upper:l3:Heinz} implies  $\beta:=\Gamma^{-1}g\in L^2[0,1]$  by using that  $\normV{g}_{\upsilon^{-2}}=\normV{[\Diag(\Gw)]^{-1}_{\um}[g]_{\um}}<\infty$. Furthermore, $\beta^m=[\Gamma]_{\um}^{-1}[g]_{\um}$ is the unique Galerkin solution of \eqref{app:unknown:Galerkin}. By using successively  the first inequality of \eqref{pr:upper:l3:Heinz},  the Galerkin condition \eqref{app:unknown:Galerkin} and  the second inequality of \eqref{pr:upper:l3:Heinz}, we obtain
\begin{equation}\label{pr:upper:l3:e1:0}
\normV{\beta-\beta^m}^2_{\Gw^2}\leqslant \Gd^2 \normV{\Gamma(\beta-\beta^m)}^2\leqslant \Gd^2 \normV{\Gamma(\beta-\Pi_m\beta)}^2\leqslant \Gd^4 \normV{\beta-\Pi_m\beta}^2_{\Gw^2}
\end{equation}
Since $(\Gw_j)$ is monotonically decreasing it follows $ \normV{\beta-\Pi_m\beta}^2_{\Gw^2} \leqslant \Gw_m^2\,\normV{\beta}^2$ and, hence
by using  \eqref{pr:upper:l3:Heinz} with $s=-1$ we have  
$ \normV{\beta-\Pi_m\beta}^2_{\Gw^2} \leqslant \Gd^2\,\Gw_m^2\,\normV{g}^2_{\Gw^{-2}}$. Combining the last estimate with \eqref{pr:upper:l3:e1:0} we obtain
\begin{equation*}
\normV{\beta^m-\Pi_m\beta}^2_{\Gw^2}\leqslant 2\{\normV{\beta-\beta^m}^2_{\Gw^2} +\normV{\beta-\Pi_m\beta}^2_{\Gw^2} \}\leqslant 2\Gd^2(\Gd^4+1)\,\Gw_m^2\,\normV{g}^2_{\Gw^{-2}}
\end{equation*}
which together with $\normV{f}^2\leqslant\Gw^{-2}_m \normV{f}^2_{\Gw^2} $ for all $f\in\Psi_m$ leads to
\begin{equation*}
\normV{\beta^m-\Pi_m\beta}^2\leqslant \Gw^{-2}_m\normV{\beta^m-\Pi_m\beta}^2_{\Gw^2} \leqslant 2\Gd^2(\Gd^4+1)\,\normV{g}^2_{\Gw^{-2}}.
\end{equation*}
By using the last estimate together with $\normV{g}_{\Gw^{-2}}=\normV{[\Diag(\Gw)]^{-1}_{\um}[g]_{\um}}$ we conclude that
\begin{multline}\label{pr:upper:l3:e1:1}
\normV{[\Gamma]_{\um}^{-1}[g]_{\um}}^2=  \normV{\beta^m}^2 \leqslant 2 \{\normV{\beta^m-\Pi_m\beta}^2 + \normV{\Pi_m\beta}^2\}\\\leqslant 2\Gd^2(2\Gd^4+3)\normV{[\Diag(\Gw)]^{-1}_{\um}[g]_{\um}}^2,\quad\forall g\in\Psi_m.
\end{multline}
Then, from \eqref{pr:upper:l3:e1:1} follows by using the inequality of \cite{Heinz51}  for all $g\in\Psi_m$
\begin{equation*}
\normV{[\Gamma]_{\um}^{-1/2}[g]_{\um}}^2\leqslant \{2\Gd^2(2\Gd^4+3)\}^{1/2} \normV{[\Diag(\Gw)]^{-1/2}_{\um}[g]_{\um}}^2,
\end{equation*}
which implies together with $\normV{[\Diag(\Gw)]^{-1}_{\um}}=\Gw_m^{-1}$ the estimate \eqref{pr:upper:l3:e1-1}, and furthermore   by replacing $[g]_{\um}$ by $[\Diag(\Gw)]^{1/2}_{\um}[g]_{\um}$ the estimate \eqref{pr:upper:l3:e1-2}, that is,
\begin{equation*}
\normV{[\Gamma]_{\um}^{-1/2}[\Diag(\Gw)]^{1/2}_{\um}[g]_{\um}}^2\leqslant \{2\Gd^2(2\Gd^4+3)\}^{1/2} \normV{[g]_{\um}}^2,\quad\forall g\in\Psi_m.
\end{equation*}

Proof of \eqref{pr:upper:l3:e1-3}.  By using the second inequality of \eqref{pr:upper:l3:Heinz} together with $\normV{\Pi_m}=1$ we obtain
\begin{equation*}
\normV{[\Gamma]_{\um}[g]_{\um}}^2= \normV{\Pi_m \Gamma g}^2\leqslant \normV{\Gamma g}^2\leqslant \Gd^2 \normV{g}^2_{\Gw^2}= \Gd^2\normV{[\Diag(\Gw)]_{\um}[g]_{\um}}^2,\quad\forall g\in\Psi_m
\end{equation*}
and hence  the inequality of \cite{Heinz51} implies 
\begin{equation*}
\normV{[\Gamma]_{\um}^{1/2}[g]_{\um}}^2\leqslant \Gd \normV{[\Diag(\Gw)]_{\um}^{1/2}[g]_{\um}}^2,\quad\forall g\in\Psi_m.
\end{equation*}
Thereby, \eqref{pr:upper:l3:e1-3} follows by replacing $[g]_{\um}$ by $[\Diag(\Gw)]^{-1/2}_{\um}[g]_{\um}$, that is,
\begin{equation*}
\normV{[\Gamma]_{\um}^{1/2}[\Diag(\Gw)]^{-1/2}_{\um}[g]_{\um}}^2\leqslant \Gd \normV{[g]_{\um}}^2,\quad\forall g\in\Psi_m.
\end{equation*}

Proof of \eqref{pr:upper:l3:e2}. Let $\beta\in \cW_\bw^\br$. Consider the decomposition
\begin{equation*}
\normV{\beta-\beta^m}_\hw^2\leqslant 2\{   \normV{\beta-\Pi_m \beta}_\hw^2 +\normV{\Pi_m \beta-\beta^m}_\hw^2\}.
\end{equation*}
Since $(\hw_j/\bw_j)$ is non-increasing it follows  $ \normV{\beta-\Pi_m\beta}_\hw^2 \leqslant \hw_m/\bw_m\,\normV{\beta}^2_\bw$, while we show below%
\begin{equation}\label{pr:upper:l3:e2:1}
\normV{\Pi_m \beta-\beta^m}_\hw^2\leqslant 2(1+\Gd^2)\, \hw_m/\bw_m\,\normV{\beta}^2_\bw.
\end{equation}
Consequently,  by combination of these two bounds the condition $\beta\in \cW_\bw^\br$, i.e., $\normV{\beta}^2_\bw\leqslant \br$, implies \eqref{pr:upper:l3:e2}. 
From \eqref{pr:upper:l3:e1:0} follows $\normV{\beta-\beta^m}^2_{\Gw^2}\leqslant \Gd^4 \normV{\beta-\Pi_m\beta}^2_{\Gw^2}\leqslant \Gd^4 \Gw_m^2/\bw_m \normV{\beta}_\bw^2$ because $(\Gw_j^2/\bw_j)$ is non-increasing,  and hence,
\begin{equation}\label{pr:upper:l3:e2:2}
\normV{\Pi_m\beta-\beta_m}^2_{\Gw^2}\leqslant 2\{ \normV{\beta-\beta^m}^2_{\Gw^2} + \normV{\beta-\Pi_m\beta}^2_{\Gw^2}\}\leqslant
2(1+\Gd^4) \Gw_m^2/\bw_m \normV{\beta}_\bw^2.
\end{equation}
Furthermore,  $ \normV{\Pi_m\beta-\beta^m}_\hw^2\leqslant \hw_m\Gw_m^{-2}\,\normV{\Pi_m\beta-\beta^m}^2_{\Gw^2}$ since $(\hw_j/\Gw_j^{2})$ is non-decreasing. The last estimate and \eqref{pr:upper:l3:e2:2} imply now together \eqref{pr:upper:l3:e2:1}, which completes the proof.\end{proof}

\subsection{Proofs of Section \ref{sec:ex}}\label{app:proofs:ex}
\paragraph{The mean prediction error.}
\begin{proof}[\noindent\textcolor{darkred}{\sc Proof of Proposition \ref{MPE:lower}.}] Since  $\Gamma \in \cN_{\Gw}^{\Gd}$, $\Gd\geqslant 1$,   it follows by using the inequality of \cite{Heinz51} that $\Ex\normV{\widetilde\beta-\beta}_\Gamma^2 \asymp_d \Ex \normV{\widetilde\beta-\beta}_\Gw^2$. Therefore, we can apply the general results  by considering the $\cW_\hw$-risk with $\hw=\Gw$  as a measure of the performance of an estimator of $\beta$. Furthermore, in case (i) the definition of $\bw_j^p$ and $\Gw_j$ imply together $(\bw_{\kstar}^p/\hw_{\kstar})\sum_{j=1}^{\kstar}  \hw_j/\Gw_j =  \kstar^{2a+2p+1}$. It follows that the condition on $\kstar$ and $\dstar$ given in \eqref{gen:def:m-gam} of Theorem \ref{gen:lower:theo} can be rewritten as $\kstar\sim n^{1/(2p+2a+1)}$ and $\dstar\sim n^{-(2p+2a)/(2p+2a+1)}$. On the other hand, in case (ii)   $(\bw_{\kstar}^p/\hw_{\kstar})\sum_{j=1}^{\kstar}  \hw_j/\Gw_j= \kstar^{2p+1}\exp(\kstar^{2a})$ implies that the condition on $\kstar$ and $\dstar$ writes  $\kstar\sim (\log n)^{1/(2a)}$ and $\dstar\sim n^{-1}(\log n)^{1/(2a)}$. Consequently, the lower bounds in Proposition \ref{MPE:lower} follow by applying Theorem \ref{gen:lower:theo}.\end{proof}

\begin{proof}[\noindent\textcolor{darkred}{\sc Proof of Proposition \ref{MPE:upper}.}]Note, that for sufficiently large $n$ 
the condition on $\gamma$ in Theorem \ref{gen:upper:theo} writes $\gamma=n$ because $(b_j^p)$ is increasing. Furthermore, it is easily seen that  the  additional condition \eqref{gen:upper:varphi:cond} is satisfied in the exponential case and   for all $k\geqslant 2+8/(2p+2a-1)$ also in the polynomial case. Finally, since in both cases  the condition on $m$ ensures that  $m\sim \kstar$ (see the proof of Proposition \ref{MPE:lower}) the result follows from Theorem \ref{gen:upper:theo}.\end{proof}

\paragraph{The estimation of derivatives.}
\begin{proof}[\noindent\textcolor{darkred}{\sc Proof of Proposition \ref{MSE:lower}.}] Since  for each $0\leqslant s \leqslant p$ we have $\Ex\normV{\widetilde\beta^{(s)}-\beta^{(s)}}^2 \sim \Ex \normV{\widetilde\beta-\beta}_{\bw^s}^2$ we can apply again the general results  by considering the $\cW_\hw$-risk with $\hw=\bw^s$. In case (i) the well-known  approximation $\sum_{j=1}^{m} j^{r}\sim m^{r+1}$ for $r>0$ together with the definition of $\bw_j^p$ and $\Gw_j$ imply  $(\bw_{\kstar}^p/\hw_{\kstar})\sum_{j=1}^{\kstar}  \hw_j/\Gw_j \sim  \kstar^{2a+2p+1}$. It follows that the condition on $\kstar$ and $\dstar$ given in \eqref{gen:def:m-gam} of Theorem \ref{gen:lower:theo} writes $\kstar\sim n^{1/(2p+2a+1)}$ and $\dstar\sim n^{-(2p-2s)/(2p+2a+1)}$. On the other hand, in case (ii) by applying Laplace's Method (c.f. chapter 3.7 in \cite{Olver1974}) the definition of $b_j$ and $\Gw_j$ 
imply   $(\bw_{\kstar}^p/\hw_{\kstar})\sum_{j=1}^{\kstar}  \hw_j/\Gw_j\sim  \kstar^{2p}\exp(\kstar^{2a})$ implies that the condition on $\kstar$ and $\dstar$ can be rewritten as   $\kstar\sim (\log n)^{1/(2a)}$ and $\dstar\sim n^{-1}(\log n)^{1/(2a)}$. Consequently, the lower bounds in Proposition \ref{MPE:lower} follow by applying Theorem \ref{gen:lower:theo}.\end{proof}

\begin{proof}[\noindent\textcolor{darkred}{\sc Proof of Proposition \ref{MSE:upper}.}]The proof follows in analogy to the proof of Proposition \ref{MPE:upper} and we omit the details.\end{proof}

\end{document}